\documentclass[12pt]{article}

\usepackage{amsmath}
\usepackage{fullpage}
\usepackage{amsfonts}
\usepackage{amssymb}
\usepackage{graphicx}
\usepackage[top=1in,bottom=1in,left=1in,right=1in]{geometry}
\newtheorem{theorem}{Theorem}

\newtheorem{lemma}{Lemma}
\newtheorem{corollary}{Corollary}
\newtheorem{remark}{Remark}
\newtheorem{definition}{Definition}
\numberwithin{equation}{section}
\newenvironment{proof}[1][Proof]{\noindent\textbf{#1.} }{\ \rule{0.5em}{0.5em}}

\renewcommand{\epsilon}{\varepsilon}
\newcommand{\ket}[1]{\mathop{\left|#1\right>}\nolimits}
\newcommand{\bra}[1]{\mathop{\left<#1\,\right|}\nolimits}

\newcommand{\kb}[2]{| #1\rangle\!\langle #2 |}
\newcommand{\Tra}[1]{\mathop{{\mathrm{Tr}}_{#1}}}
\newcommand{\Tr}[2]{\mathop{{\mathrm{Tr}}_{#1}} (#2) }

\def\N{\mathcal{N}}

\def\M{\mathcal{M}}
\def\T{\mathcal{T}}

\def\C{\mathcal{C}}

\def\D{\mathcal{D}}

\def\U{\mathcal{U}}

\makeatletter
\let\@copyrightspace\relax
\makeatother

\begin{document}

\sloppy

\centerline{\Large \bf Quantum Information Complexity and}
\centerline{\Large \bf Amortized Communication}

\centerline{Dave Touchette \,\footnote{touchette.dave@gmail.com, Laboratoire d'informatique th\'eorique et quantique, D\'epartement d'informatique et de recherche op\'erationnelle, Universit\'e de Montr\'eal. }}

\thispagestyle{empty}

\begin{abstract}

We define a new notion of information cost for quantum protocols, and a 
corresponding notion of quantum information complexity for bipartite 
quantum channels, and then investigate the properties of such quantities.
These are the fully quantum generalizations of the analogous quantities for bipartite
classical functions that have found many applications recently, in particular for proving communication complexity lower bounds.
Our definition is 
strongly tied to the quantum state redistribution task.

Previous attempts have been made to define such a
quantity for quantum protocols, with particular applications in mind; our notion
differs from 
these in many respects.
First, it directly provides a lower bound on the quantum 
communication cost, independent of the number of rounds of the 
underlying protocol. Secondly,
we provide an operational interpretation for
quantum information complexity: we show that it is exactly equal
to the amortized quantum communication complexity of a bipartite 
channel on a given state. 
This generalizes a result of Braverman and Rao to quantum protocols,
and even strengthens the classical result in a bounded round scenario.
Also, this provides an analogue of the Schumacher 
source compression theorem for interactive quantum protocols, and
answers a question raised by Braverman.

We also discuss 
some potential applications to quantum communication complexity 
lower bounds by specializing our definition for classical functions and inputs.
Building on work of Jain, Radhakrishnan and Sen, we provide new evidence suggesting that the
bounded round quantum communication complexity of the disjointness function is $\Omega ( \frac{n}{M} + M)$, for $M$-message protocols.
This would match the best known upper bound.
\end{abstract}




\newpage

\setcounter{page}{1}

\section{Introduction}

The paradigm of information complexity has been quite successful recently
 in classical communication complexity. What started out as a 
useful tool for proving communication complexity lower bounds has 
recently developed into an important subfield of its own.
The definition of information complexity, the sum 
of the mutual information between the protocol transcript and 
each player's input conditional on the other player's input,
makes it possible to bring powerful tools from information theory to study interactive 
communication. Many recent results show 
that this paradigm has enabled researchers to tackle questions that seemed out of reach not so long ago,
like an exact characterization by Braverman, Garg, Pankratov, and Weinstein \cite{BGPW13} of the communication complexity of the disjointness function.

The first results on information complexity were those of Chakrabarti, Shi, Wirth and Yao \cite{CSWY01}, 
who gave lower bounds for communication complexity in the simultaneous
 message passing model by defining what would now be called external information
 cost. A subsequent paper by Bar-Yossef, Jayram, Kumar and Sivakumar \cite{BYJKS02} implicitly defined what is now 
called the internal information cost, or simply information cost. They
used it to prove communication complexity lower bounds for 
functions that can be built from simpler component functions, 
like the set disjointness function with respect to the AND function. The recent wave of results started 
 with a compression result of Barak, Braverman, Chen and Rao \cite{BBCR10} that had application in particular
to proving a direct sum theorem for randomized communication complexity. 
In \cite{BR11}, Braverman and Rao give an exact operational interpretation of the 
information complexity as the  amortized distributional communication 
complexity, which can be viewed as an analogue of the 
Shannon source compression theorem for interactive protocols. 
In \cite{Bra12a}, Braverman provides a prior-free definition for information complexity of 
classical functions, and presents a similar operational interpretation 
for worst case amortized communication complexity. 
Since then, many results have been derived  by using the  information complexity 
paradigm~\cite{BGPW13, BM13, BEOPV13}, witnessing its power.
Braverman provides a nice overview of results circa 2012 \cite{Bra12b}.

It might appear surprising that the definition of this simple 
looking quantity, the information cost of a protocol, has such far 
reaching consequences. It is tempting to try to bring such 
a paradigm to the quantum setting and hopefully get comparable results. 
However, there are many difficulties in trying to get 
such a quantum generalization of information cost.
Firstly, due to monogamy of entanglement, 
there is no good notion of an overall transcript and its corresponding correlation with the
localized inputs at the beginning of the protocol. Indeed, it is only possible to evaluate 
quantum information quantities for quantum registers that 
are defined at the same moment in time. But then, the no-cloning 
theorem \cite{Dieks82, WZ82} forbids copying of previous messages to generate a final
 transcript accounting for all communication in the quantum protocol, 
on which we could evaluate a quantum information cost. 
Even for protocols with pre-shared entanglement and classical communication, 
the final classical transcript could be completely uncorrelated to the inputs. This can be seen for example in 
a protocol that teleports at each time step, which would
generate a transcript that is uniformly distributed.
However, 
notwithstanding these difficulties, Jain, Radhakrishnan and Sen \cite{JRS03}, as well as Jain and Nayak \cite{JN13}, gave
different definitions for such a quantity with specific applications in mind. 
In particular, the definition in \cite{JRS03} leads to a beautiful proof 
of a lower bound of $\Omega (\frac{n}{ M^2} + M)$ on bounded round quantum communication 
complexity for (size $n$) set disjointness computed by 
$M$-message protocols. Note that the remark is made in \cite{JRS03} that the optimal $\Theta ( \sqrt{n}) $ protocol of Aaronson and Ambainis \cite{AA03} can be adapted to yield a bounded round protocol achieving communication of $O(\frac{n}{M} + M)$. Hence, this comes close to match the best known 
upper bound.
However, even if these definitions can be successful for obtaining 
interesting results in a bounded round scenario, it is quite plausible 
that they are also limited to such applications. Indeed,
given these previous definitions of quantum information cost, it is quite easy
 to find particular inputs and protocols with $M$ messages and 
quantum communication cost $C$ such that the quantum 
information cost is $\Omega (M \cdot C)$. We believe that a more natural
definition of quantum information cost would not have such a dependence 
on the round complexity of protocols. 
In contrast, the classical information cost directly provides a lower bound on the communication cost,
and this property is crucial in many of its recent applications.
Hopefully, it would also lead to an operational
 interpretation analogous to the classical one for amortized quantum 
communication complexity. A related question 
was asked by Braverman \cite{Bra12a}, on how to define the correct quantum analogue 
of information complexity.

\section{Overview of Results}

We show that despite the previously mentioned setbacks, it is possible 
to define a quantum information cost quantity for protocols 
and a corresponding quantum information complexity quantity 
for channels that satisfies the two important properties asked for above.
First, they are lower bounds on the quantum communication cost and 
quantum communication complexity of the corresponding protocols 
and channels, respectively. Secondly, we can compress asymptotically 
any protocol to its quantum information cost, and so we can give 
the operational interpretation for quantum information complexity 
as the amortized quantum communication complexity. That is, the 
quantum communication complexity per copy for implementing $n$ 
copies of the channel, in the asymptotic limit of large $n$. 
The quantity that we consider would be the fully quantum analogue of the 
distributional information complexity of functions: we study the quantum 
information complexity of a channel on some arbitrary quantum input.
In the case of a so-called classical input, this corresponds to a probability 
distribution. When implementing a bipartite channel with a protocol 
on a general input, we allow for some small error $\epsilon$ in the 
trace distance with respect to a reference system purifying the input.
The introduction of the reference system is to ensure 
that our implementation maintains correlations with the outside 
world as well as the actual channel. For classical inputs and 
functions, we show that this implies a bound on the probability of failure, in the sense
of a classical average on the input distribution, of 
the quantum protocol for computing the function. Hence, this links our quantity to the 
distributional quantum communication complexity of 
classical functions. When implementing $n$ copies of a
 general channel, the success criterion is that for each 
instance of the channel, the error is bounded by $\epsilon$.

To circumvent the difficulties we mentionned in defining a quantum
 generalization of information cost and complexity, we take 
a different perspective on the classical definitions. Indeed, 
our main insight in defining such quantities comes by using
 a rewriting of the information cost of a protocol that was 
already implicit in previous works on information complexity. 
We reinterpret this quantity by viewing each message generation 
in the protocol as a noisy channel whose output is to be sent over a
noiseless channel. Then, message transmission
 is viewed as the simulation, over a noiseless channel, of a noisy channel with feedback to the
 sender and side information at the receiver. This is a variant of
 what is known as a (tensor power input) classical reverse Shannon theorem in the information 
theory literature. Using previously known results 
\cite{BSST02, LD09}, we obtain an alternate, simple proof 
of the operational interpretation of classical information complexity. 
An additional feature of this proof in contrast to the one of Braverman and Rao \cite{BR11}
is that it maintains the round complexity of the original protocol, since
these results about noisy channel simulation use only unidirectional transmission.
To the best of our knowledge, this is the first proof to establish that the $M$-message information
complexity is exactly equal to the $M$-message amortized communication complexity.
Note however that this proof only works in the asymptotic limit, so
we do not get as a bonus of the techniques used for our coding theorem a compression theorem for
a single copy of the protocol. It is reasonable to hope that 
eventual results on so-called one-shot unidirectional coding theorems might lead to interesting results about
bounded round, single copy protocol compression.

 Using this perspective on classical information cost,
we can more easily define a quantum generalization.
The right quantum analogue of the reverse Shannon task with
 feedback and side information is quantum state redistribution, 
for which Devetak and Yard give an optimal protocol \cite{DY08, YD09}. 
With these tools in hands, we then
define the quantum analogues of information cost and complexity,
 and prove their operational interpretations. In addition, we also prove some
 interesting properties of these quantities.
In particular, it is an almost immediate consequence of the definition
 that the quantum information cost of a protocol lower bounds its
 quantum communication cost, and then a similar result holds for
quantum communication complexity relative to channels.
We also prove that with our definition, the quantum information 
complexity is 
an additive quantity, and that it is convex and continuous in the channel and error parameter. Finally, relative to input states, we prove
a concavity result on quantum information cost.
At this point, given this operational interpretation and the 
fact that information is a lower bound on communication, 
we might argue that this corresponds to the right quantum analogue 
of the information complexity. 

To demonstrate the potential of the quantum information
 complexity paradigm, we reduce the quantum information complexity of the disjointness
 function to that of the AND function, 
a reduction along similar lines as that of \cite{BYJKS02, JRS03}. In contrast to the quantum reduction of \cite{JRS03},
 we are able to get rid of a factor $\frac{1}{M}$ for quantum protocols with $M$ messages.
Given that their lower bound for 
bounded round complexity for $n$-bit disjointness is $\Omega (\max (\frac{n}{M^2},M))$, 
we are thus optimistic that this reduction by a factor of $\frac{1}{M}$ 
using our approach could  lead to a lower bound 
of $\Omega (\max (\frac{n}{M}, M))$, thus matching the best currently
 known upper bound \cite{AA03, JRS03}. 
 Like classical information
 cost, our quantity is defined in terms of conditional mutual
 information, a quantity that has been notoriously hard to lower
 bound in the quantum setting \cite{LR73, BCY11}. 
We leave as an interesting open problem to develop tools for
lower bounding the conditional mutual information in our setting, and to see if we can obtain such
a result for the disjointness function. We explore
 other potential applications in the conclusion.

\paragraph{Organization.}
The structure of the paper is the following. In the next section, 
we fix the notation that we use for quantum mechanics, present the 
necessary quantum information theory background, define 
formally the quantum communication model that we use and also define 
quantum communication complexity in this model. We then present a 
perspective on classical information complexity that leads us to a 
quantum generalization, then give such a definition of quantum 
information complexity and finally prove its operational interpretation in 
the following section. We also explore some of the properties of our 
definition, and then go on to discuss some potential applications. 
We conclude with a discussion of our results, additional potential applications, and further research directions.

\section{Preliminaries}
	\label{sec:prel}

\subsection{Quantum Information Theory}

We use the following notation for quantum theory; see \cite{Wat13, Wilde11} for more details. 
We associate a quantum register $A$ with
a corresponding vector space, also denoted by $A$. We only consider 
finite-dimensional vector spaces. A state of quantum register $A$ is 
represented by a density operator $\rho \in \D (A)$, with $\D (A)$ the set 
of all unit trace, positive semi-definite linear operators mapping $A$ into itself. 
We say that a state $\rho$ is pure if
it is a projection operator, i.e.~$(\rho^{AR})^2 = \rho^{AR} $. 
For a pure state $\rho$, we often use the pure state formalism, and represent $\rho$ by the
vector $\ket{\rho}$ it projects upon, i.e.~$\rho = \kb{\rho}{\rho}$.
A quantum channel from quantum register $A$ into quantum register $B$ is represented
by a super-operator $\N^{A \rightarrow B} \in \C (A, B)$, with $\C (A, B ) $ the set of
 all completely positive, trace-preserving linear operators from $\D( A )$ 
into $D(B)$. If $A=B$, we might simply write $\N^A$, and when systems are clear from context, we might drop the superscripts.
For channels $\N_1 \in \C (A, B), \N_2 \in \C (B, C)$ and state $\rho \in \D (A)$, we denote their composition as $\N_2 \circ \N_1 \in \C (A, C)$, with action $\N_2 \circ \N_1 (\rho) = \N_2 (\N_1 (\rho))$. We might drop the $\circ$ if the composition is clear from context.
For $A$ and $B$ isomorphic, we denote the identity mapping 
as $I^{A \rightarrow B}$, with some implicit choice for the change of basis.
 For $\N^{A_1 \rightarrow B_1} \otimes I^{A_2 \rightarrow B_2} \in C(A_1 \otimes A_2, B_1 \otimes B_2)$, 
we might abbreviate this as $\N$ and leave the identity channel implicit 
when the meaning is clear from context. An important subset
 of $\C (A, B)$ when $A$ and $B$ are isomorphic spaces is the set
 of unitary channels $\U (A, B)$, the set of all maps $U \in\C (A, B)$ with an 
adjoint map $U^\dagger \in \C (B, A)$ such that
 $U^\dagger \circ U = I^{A}$. 
Another important example of channel that we use is the
 partial trace $\Tr{B}{\cdot} \in C(A \otimes B, A)$ which effectively 
gets rid of the $B$ subsystem. 
Fixing a basis $\{ \ket{b} \}$ for $B$, the action of $\Tra{B}$ on 
any $\rho^{AB} \in D(A \otimes B)$ is $\Tr{B}{\rho^{AB}} = 
\sum_b \bra{b} \rho^{AB} \ket{b}$, and we write $\rho^A = \Tr{B}{\rho^{AB}}$.
We also denote $\Tra{\neg A} = \Tra{B}$ to express that we want to keep only the $A$ register.
Fixing a basis also allows us to talk about classical states and 
joint states: $\rho \in D(B)$ is classical (with respect to this basis) if it is diagonal in basis
 $\{ \ket{b} \}$, i.e.~$\rho = \sum_b p_B (b) \kb{b}{b}$ for some probability
 distribution $p_B$. More generally, subsystem $B$ of $\rho^{AB}$ is 
said to be classical if we can write $\rho^{AB} = \sum _b p_B (b) 
\kb{b}{b}^B \otimes \rho_b^A$ for some $\rho_b^A \in D(A)$.
An important example of a channel mapping a quantum system 
to a classical one is the measurement channel $\Delta_B$, defined 
as $\Delta_B ( \rho) = \sum_b \bra{b} \rho \ket{b} \cdot \kb{b}{b}^B$ for any $\rho \in D(B)$.
Often, $A, B, C, \cdots$ will be used to discuss general systems,
 while $X, Y, Z, \cdots$ will be reserved for classical systems.
For a state $\rho^A \in D (A)$, a purification is a pure state
 $\rho^{AR} \in \D (A \otimes R)$ satisfying $\Tr{R}{\rho^{AR}} = \rho^A$. 
If $R$ has dimension at least that of $A$, then such a purification always exists.
For a given $R$, all purifications are equivalent up to unitaries.
For a channel $\N \in C(A, B)$, a unitary extension is a
 unitary $U_\N \in U(A \otimes B^\prime, A^\prime \otimes B)$ with
 $\Tr{A^\prime}{U_\N (\rho^A \otimes \sigma^{B^\prime})} 
= \N (\rho^A)$ for some fixed $\sigma \in \D (B^\prime)$.
It is sufficient to consider any fixed pure state $\sigma$.
Such an extension always
 exists provided $A^\prime$ is of dimension at least $\dim (A)^2$ (note that we also must have 
$\dim (A) \cdot \dim (B^\prime) = \dim(A^\prime) \cdot \dim(B)$).

The notion of distance we use is the trace distance, defined for 
two states $\rho_1, \rho_2 \in D (A)$ as the sum of the absolute values of the eigenvalues of their difference:
\begin{align*}
 \| \rho_1 - \rho_2 \|_A = \Tr{}{| \rho_1 - \rho_2 |}.
\end{align*}
It has an operational interpretation as four times the best bias possible in a 
state discrimination test between $\rho_1$ and $\rho_2$.
The subscript tells on which subsystems the trace distance is evaluated, and remaining subsystems
might need to be traced out.
We use the following results about trace distance. For proofs of 
these and other standard results in quantum information theory that 
we use, see \cite{Wilde11}. The trace distance is monotone under
 noisy channels: for any $\rho_1, \rho_2 \in \D (A)$ and $\N \in C(A, B)$,
\begin{align}
	\| \N (\rho_1) - \N (\rho_2) \|_B \leq \| \rho_1 - \rho_2 \|_A.
\end{align}
For unitaries, the equality becomes an identity, a property called unitary
 invariance of the trace distance. Hence, for any $\rho_1, \rho_2 \in D(A)$ and any $U \in \U (A, B)$, we have
\begin{align}
	\| U (\rho_1) - U (\rho_2) \|_B = \| \rho_1 - \rho_2 \|_A.
\end{align}
Also, the trace distance cannot be increased by adjoining an uncorrelated system:
for any $\rho_1, \rho_2 \in D(A), \sigma \in \D (B)$
\begin{align}
	\| \rho_1 \otimes \sigma - \rho_2 \otimes \sigma \|_{AB} =  \| \rho_{1} - \rho_{2}  \|_{A}.
\end{align}
It follows that the trace distance obeys a property that we call joint linearity: 
for a classical system $X$ and two states $\rho_1^{XA} = 
p_X (x) \kb{x}{x}^X \otimes \rho_{1, x}^A, \rho_2^{XA} = p_X (x) \kb{x}{x}^X \otimes \rho_{2, x}^A$,
 \begin{align}
	\| \rho_1 - \rho_2 \|_{XA} =  \sum_x p_X (x) \| \rho_{1,x} - \rho_{2, x}  \|_{A}.
\end{align}

The measure of information that we use is the von Neumann entropy, defined for any state $\rho \in D(A)$ as
\begin{align*}
	H(A)_\rho = \Tr{}{\rho \log \rho},
\end{align*}
in which we take the convention that $0 \log 0 = 0$, justified by a continuity argument. All logarithms are taken base $2$. 
Note that $H$ is invariant under unitaries applied on $\rho$.
If the 
state to be evaluated is clear from context, we might drop the subscript. 
Conditional entropy for a state $\rho^{ABC} \in D(A \otimes B \otimes C)$ is then defined as
\begin{align*}
	H(A | B)_{\rho^{AB}} = H(AB)_{\rho^{AB}} - H(B)_{\rho^B},
\end{align*}
mutual information as
\begin{align*}
	I(A; B )_{\rho^{AB}} = H (A)_{\rho^A} - H(A | B)_{\rho^{AB}},
\end{align*}
and conditional mutual information as
\begin{align*}
	I(A; B | C )_{\rho^{ABC}} = H (A | C)_{\rho^{AC}} - H(A | B C)_{\rho^{ABC}}.
\end{align*}
Note that mutual information and conditional mutual information are symmetric in interchange of $A, B$.
For any pure bipartite state $\rho^{AB} \in D(A \otimes B)$, the entropy on each subsystem is the same:
\begin{align}
	H(A) = H(B).
\end{align}
For isomorphic $A, A^\prime$, a maximally entangled state $\psi \in \D (A \otimes A^\prime)$ is
a pure state satisfying $H(A) = \log \dim (A)$.
For a system $A$ of dimension $\dim (A)$ and any $\rho \in D(A \otimes B \otimes C )$, we have the bounds
\begin{align}
	0 \leq H(A) \leq \log \dim (A), \\
	- H (A) \leq H(A | B) \leq H (A), \\
	0 \leq I(A; B) \leq 2 H (A), \\
	0 \leq I(A; B | C) \leq 2 H (A).
\end{align}
The conditional mutual information satisfy a chain rule: for any $\rho \in D(A \otimes B \otimes C \otimes D)$,
\begin{align}
	I (AB ; C | D) = I (A; C | D) + I(B; C | A D).
\end{align}
For product states $\rho^{A_1 B_1 C_1 A_2 B_2 C_2} = \rho_1^{A_1 B_1 C_1} \otimes \rho_2^{A_2 B_2 C_2}$, entropy is additive,
\begin{align}
	H(A_1 A_2) = H(A_1) + H(A_2),
\end{align}
and so there is no conditional mutual information between product system,
\begin{align}
	I(A_1 ; A_2 | B_1 B_2 ) = 0,
\end{align}
and conditioning on a product system is useless,
\begin{align}
	I(A_1 ; B_1 | C_1 A_2 ) = I(A_1 ; B_1 | C_1 ).
\end{align}
More generally,
\begin{align}
	I(A_1 A_2 ; B_1 B_2 | C_1 C_2 ) = I(A_1 ; B_1 | C_1 ) + I(A_2; B_2 | C_2).
\end{align}
Two important properties of the conditional mutual information are strong subadditivity and the data processing inequality:
we consider an equivalent rewriting of strong subadditivity, which states that conditional mutual information is non-negative.
For any $\rho \in \D (A \otimes B \otimes C)$ and $\N \in \C (B, B^\prime)$, with $\sigma = \N (\rho)$,
\begin{align}
	I (A; B | C)_\rho & \geq 0, \\
	I (A; B | C)_\rho & \geq I (A; B^\prime | C)_\sigma.
\end{align}
For classical systems, conditioning is equivalent to taking an average:
for any $\rho^{ABCX} = \sum_x p_X(x) \kb{x}{x}^X \otimes \rho_x^{ABC}$, for a classical system $X$ and
some appropriate $\rho_x \in \D (A \otimes B \otimes C)$,
\begin{align}
	H (A | B X )_\rho & = \sum_x p_X (x) H (A | B)_{\rho_x},\\
	I (A; B | C X )_\rho & = \sum_x p_X (x) I(A; B | C)_{\rho_x}.
\end{align}

\subsection{Quantum Communication Model}
\label{sec:qucomm}

We want to study in full generality the quantum communication complexity 
of bipartite quantum channels on particular input states. This is the 
generalization of distributional communication complexity of classical functions to
the fully quantum setting, 
and contains as a special case the distributional 
quantum communication complexity of classical functions. The model for 
communication complexity that we consider is the following. For a given 
bipartite channel $\N \in \C(A_{in} \otimes B_{in}, A_{out} \otimes B_{out})$ 
and input state $\rho \in \D (A_{in} \otimes B_{in})$, Alice and Bob are 
given input registers $A_{in}, B_{in}$ at the outset of the protocol, respectively, 
and they output registers $A_{out}, B_{out}$ at the end of the protocol, 
respectively, which should be in state $N(\rho)$. We generally allow for some small
error $\epsilon$ in the output, which will be formalized below. 
In the usual communication complexity setting, the input is be a classical state $\rho = \sum_{x, y} p_{XY} (x, y) \kb{x}{x}^{A_{in}} \otimes \kb{y}{y}^{B_{in}}$, the channel $\N$ implements a classical functions $\N (\kb{x}{x} \otimes \kb{y}{y}) = \kb{f_A (x, y)}{f_A (x, y)}^{A_{out}} \otimes \kb{f_B (x, y)}{f_B (x, y)}^{B_{out}}$, and the error parameter is related to the probability of failure $\sum_{x, y} p_{XY} (x, y) [\Pi (x, y) \not= (f_A(x, y), f_B (x, y)) ] \leq \frac{\epsilon}{2}$, as proved in section \ref{sec:cltd}.

A protocol $\Pi$ for implementing $\N$ on input $\rho^{A_{in} B_{in}}$ is 
defined by a sequence of unitaries $U_1, \cdots, U_{M + 1}$ along with a 
pure state $\psi \in \D (T_A \otimes T_B)$ shared between Alice and Bob, 
for arbitrary finite dimensional registers $T_A, T_B$. For appropriate finite 
dimensional memory registers $A_1, A_3, \cdots A_{M - 1}, A^\prime$ held by Alice, $B_2, B_4, \cdots B_{M - 2}, B^\prime$ 
held by Bob, and communication registers $C_1, C_2, C_3, \cdots C_M$ 
exchanged by Alice and Bob, we have (see Figure \ref{fig:prot}) 
$U_1 \in \U(A_{in} \otimes T_A, A_1 \otimes C_1), 
U_2 \in \U(B_{in} \otimes T_B \otimes C_1, B_2 \otimes C_2), 
U_3 \in \U(A_1 \otimes C_2, A_3 \otimes C_3), 
U_4 \in \U(B_2 \otimes C_3, B_4 \otimes C_4), \cdots , 
U_{M} \in \U(B_{M - 2} \otimes C_{M - 1}, B_{out} \otimes B^\prime \otimes C_{M}), 
U_{M  + 1} \in \U(A_{M - 1} \otimes C_M, A_{out} \otimes A^\prime)$.
We slightly abuse notation and also write $\Pi$ to denote the channel implemented by the protocol, i.e.
\begin{align}
\Pi (\rho) =
\Tr{A^\prime B^\prime }{U_{M  + 1} U_M \cdots U_2 U_1 (\rho \otimes \psi)}.
\end{align}
Then we say that a protocol $\Pi$ for implementing channel $\N$ 
on input $\rho^{A_{in} B_{in}}$, with purification $\rho^{A_{in} B_{in} R}$ 
for a reference system $R$, 
has error $\epsilon \in [0, 2]$ if 
\begin{align}
|| \Pi (\rho) - \N (\rho) ||_{A_{out} B_{out} R} \leq \epsilon.
\end{align}
We denote the set of all such protocol as $ \T (\N, \rho, \epsilon)$.
If we want to restrict this set to bounded round protocols with  $M$ messages, we write $\T^M (\N, \rho, \epsilon)$.
Note that for simplicity, we only define protocols with an even number of messages; our results also hold without this restriction, though in the
special case of one round protocols, we would rather consider bipartite channels with a single output to ensure that the quantum communication complexity is well-defined.
The introduction of the reference system $R$ is essential to ensure that the protocol preserves any correlation
the input state might have with the outside world as well as the channel it is supposed to implement.
As said before, for classical functions on classical input distributions, we prove a lemma in section \ref{sec:cltd} that relates this
to the probability of failure of the protocol on such a distribution.

      \begin{figure}

         \centering
      			  \includegraphics[width=12cm]{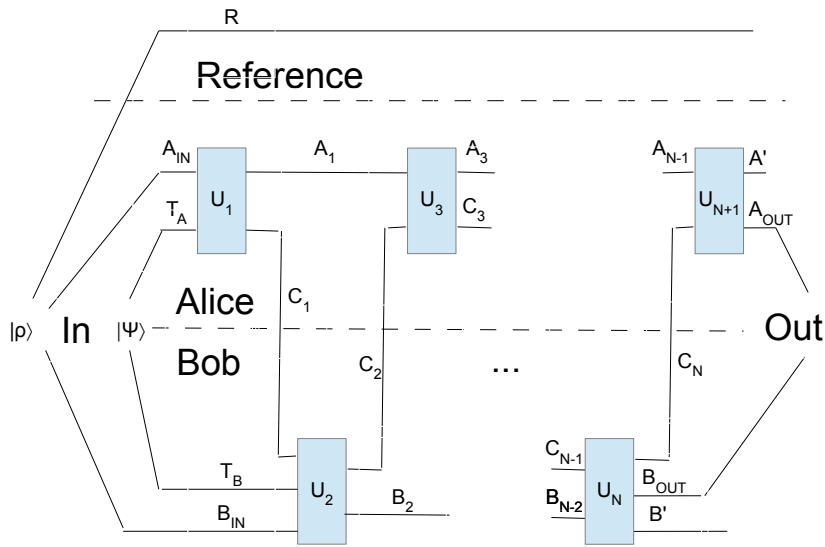}
	\caption{A standard protocol in our quantum communication model}
	\label{fig:prot}
	\end{figure}

Note that in the standard context of quantum communication complexity, 
our model would be akin to the model introduced by Cleve and Buhrman \cite{CB97}, with pre
-shared entanglement (though the fact that we use quantum 
communication here instead of classical communication as in 
the original model could lead to an improvement up to a factor 
of two of the communication complexity, due to superdense coding \cite{BW92}, but no more, due to 
the teleportation protocol \cite{BBCJPW93}), rather than to the model introduced by Yao \cite{Yao93}, in which 
parties locally initialize their registers. This is the natural analogue 
of the framework for classical information complexity in which parties 
are allowed shared randomness for free, and this seems to be
 necessary to obtain the operational interpretation of information 
complexity, classical and quantum, as the amortized communication
 complexity. Known proofs of the additivity property rely heavily 
on the availability of shared resources to perform some kind of simulation.
This is also true of many other interesting properties of information complexity.
Even though additivity and other results might not hold in the Yao model, studying information complexity
in this model might still make sense: we would first have to restrict ourselves to protocols in which the pre-shared state
$\psi$ is a pure product state. However, the definition of quantum information complexity would need to be somewhat modified,
since the definition we give in section \ref{sec:qic} allows for entanglement distribution at no cost, which is consistent with our Cleve-Buhrman like model of communication.

As was said before, our framework is the quantum generalization 
of the one for distributional information complexity, and so let us formally 
define the different quantities that we work with.

\begin{definition}
For a protocol $\Pi$ as defined above, we define the \emph{quantum communication cost} of $\Pi$ as
\begin{align*}
	QCC (\Pi) = \sum_{i} \log \dim (C_i).
\end{align*}
\end{definition}

Note that we do not require that $\dim (C_i) = 2^k$ for some $k \in \mathbb{N}$, as is usually done.
This will not affect our definition on information cost and complexity, nor on amortized communication complexity, but
might affect the single-copy quantum communication complexity by at most a factor of two.
The corresponding notion of quantum communication complexity of a channel is:

\begin{definition}
For a bipartite channel $\N \in \C (A_{in} \otimes B_{in}, A_{out} 
\otimes B_{out})$, an input state $\rho \in D(A_{in} \otimes B_{in})$ and an
error parameter $\epsilon \in [0, 2]$, we define the $\epsilon$-error \emph{quantum communication complexity} of $\N$ 
on input $\rho$ as
\begin{align*}
	QCC (\N, \rho, \epsilon) = \min_{\Pi \in \T (\N, \rho, \epsilon)} QCC (\Pi).
\end{align*}
\end{definition}

Note that this quantity is discontinuous in its parameters. Also note that no good bound is known on the size of the
entangled state that might be required to achieve this minimum. See \cite{MV14} for
a recent discussion on related issues in a different setting.
We make the following trivial remark that quantum communication complexity is decreasing in the error parameter,
that it vanishes for $\epsilon = 2$, that it is bounded by $\log \dim (A_{in}) +\log \dim ( A_{out} )$, and that it also 
vanishes for any pure state $\rho$.
At $\epsilon = 2$, it is because the trace distance is saturated at $2$ and so we can consider
a protocol that outputs anything without communication, while for pure 
states it is because there is no correlation with the outside world, so we can 
consider a protocol that is given, as entanglement for the protocol, the output of the channel acting on the pure state, and outputs it without communication.

\begin{remark} 
For any $\N, \rho, 0 \leq \epsilon_1 \leq \epsilon_2 \leq 2$, the following holds:
\begin{align*}
QCC(\N, \rho, \epsilon_2) & \leq QCC(\N, \rho, \epsilon_1), \\
QCC (\N, \rho, 0) & \leq \log \dim (A_{in}) + \log \dim (A_{out}), \\
QCC(\N, \rho, 2) & = 0.
\end{align*}
Also, for any $\N, \epsilon \in [0,2]$, the following holds for any pure state $\rho$ :
\begin{align*}
QCC(\N, \rho, \epsilon) = 0.
\end{align*}
\end{remark}

We have the following definition for bounded round quantum communication complexity, and
similar remarks hold.

\begin{definition}
For a bipartite channel $\N \in \C (A_{in} \otimes B_{in}, A_{out} 
\otimes B_{out})$, an input state $\rho \in D(A_{in} \otimes B_{in})$,an
error parameter $\epsilon \in [0, 2]$ and a bound $M \in \mathbb{N}$ on the number of messages, 
we define the $M$-message, $\epsilon$-error \emph{quantum communication complexity} of $\N$ 
on input $\rho$ as
\begin{align*}
	QCC^M (\N, \rho, \epsilon) = \min_{\Pi \in \T^M (\N, \rho, \epsilon)} QCC (\Pi).
\end{align*}
\end{definition}

We are also interested in the amortized quantum communication 
complexity of channels. A protocol $\Pi_n$ is said to compute the 
$n$-fold product channel $\N^{\otimes n}$ on input $(\rho)^{\otimes n}$ with error 
$\epsilon$ if for all $i \in [n]$, 
\begin{align}
\| \Tra{\neg (A_{in}^i B_{in}^i R^i)}
\circ \Pi_n (\rho^{\otimes n}) - 
\N (\rho) \|_{A_{out}^i B_{out}^i R^i} \leq \epsilon.
\end{align} 
We have implicitly used the fact that it is possible to find a purification of $\rho^{\otimes n}$
with a decomposition of the purifying register $R = R_1 \otimes \cdots \otimes R_n$, and with the $i$-th copy of $\rho$
purified by the reference subregister $R_i$.
This error criterion corresponds to the one achieved when sequentially simulating $n$ times channel $\N$ on input
 $\rho$, each time with error $\epsilon$, and is weaker 
than demanding to simulate it $n$ times with overall error $\epsilon$.
 The reason for this is that asking for overall error $\epsilon$ 
could be a much harder task. Indeed, consider a purified 
input state that is $\epsilon$ away in trace distance to a state 
which is product with respect to the $A_{in} B_{in} - R$ bipartite cut. 
Then, since the trace distance is monotone under noisy channels,
 the parties can simulate the channel at zero communication cost
 and achieve error $\epsilon$ by taking,
as part of the entanglement of their protocol,
 the $A_{out} B_{out}$ registers
 of the channel acting on that product state. 
Thus the quantum information complexity is also zero. We can then
 also achieve the task of amortized quantum communication 
complexity with $\epsilon$ error in each input at zero communication. 
However, using the operational interpretation of the trace distance as the best bias
in a distinguishability experiment, the amortized quantum communication task in which we ask 
for overall error $\epsilon$ cannot be achieved at zero communication
 cost, since having access to many instances of the output state allows
 for better distinguishability whenever starting with distinguishability greater 
than zero between the actual input and the product state \cite{BCGST02}. Hence, if we want to obtain the intended 
operational interpretation, we have to settle for such a success 
parameter. We denote $\T_n (\N^{\otimes n}, \rho^{\otimes n}, \epsilon)$ 
the set of all protocols achieving the above goal of having $\epsilon$ error 
in each output, and can define the n-fold quantum communication complexity accordingly.

\begin{definition}
For a bipartite channel $\N \in \C (A_{in} \otimes B_{in}, A_{out} 
\otimes B_{out})$, an input state $\rho \in D(A_{in} \otimes B_{in})$ and an
error parameter $\epsilon \in [0, 2]$, we define the $\epsilon$-error,
\emph{$n$-fold quantum communication complexity} of $\N$ on input $\rho$ as
\begin{align*}
	QCC_n (\N^{\otimes n}, \rho^{\otimes n}, \epsilon) = \min_{\Pi_n  \in \T_n (\N^{\otimes n}, \rho^{\otimes n}, \epsilon)} QCC (\Pi_n).
\end{align*}
\end{definition}

\begin{definition}
For a bipartite channel $\N \in \C (A_{in} \otimes B_{in}, A_{out} 
\otimes B_{out})$, an input state $\rho \in D(A_{in} \otimes B_{in})$, 
and an error parameter $\epsilon \in [0, 2]$, we define the $\epsilon$-error \emph{amortized 
quantum communication complexity} of $\N$ on input $\rho$ as
\begin{align*}
	AQCC (\N, \rho, \epsilon) =  \limsup_{n \rightarrow \infty} \frac{1}{n} QCC_n (\N^{\otimes n}, \rho^{\otimes n}, \epsilon).
\end{align*}
\end{definition}

Note that for all $n$, $QCC_n (\N^{\otimes n}, \rho^{\otimes n}, \epsilon) \leq n QCC (\N, \rho, \epsilon)$, as is made clear by running $n$ times in parallel a protocol achieving the minimum in the definition of the quantum communication complexity. Hence, the amortized quantum communication complexity is bounded by the quantum communication complexity.

We have corresponding definitions for bounded round complexity.

\begin{definition}
For a bipartite channel $\N \in \C (A_{in} \otimes B_{in}, A_{out} 
\otimes B_{out})$, an input state $\rho \in D(A_{in} \otimes B_{in})$, an
error parameter $\epsilon \in [0, 2]$ and a bound $M \in \mathbb{N}$ on the number of messages, we define the $M$-message, $\epsilon$-error,
\emph{$n$-fold quantum communication complexity} of $\N$ on input $\rho$ as
\begin{align*}
	QCC_n^M (\N^{\otimes n}, \rho^{\otimes n}, \epsilon) = \min_{\Pi_n  \in \T_n^M (\N^{\otimes n}, \rho^{\otimes n}, \epsilon)} QCC (\Pi_n).
\end{align*}
\end{definition}

\begin{definition}
For a bipartite channel $\N \in \C (A_{in} \otimes B_{in}, A_{out} 
\otimes B_{out})$, an input state $\rho \in D(A_{in} \otimes B_{in})$, 
an error parameter $\epsilon \in [0, 2]$ and a bound $M \in \mathbb{N}$ on the number of messages, we define the $M$-message, $\epsilon$-error \emph{amortized 
quantum communication complexity} of $\N$ on input $\rho$ as
\begin{align*}
	AQCC^M (\N, \rho, \epsilon) =  \limsup_{n \rightarrow \infty} \frac{1}{n} QCC_n^M (\N^{\otimes n}, \rho^{\otimes n}, \epsilon).
\end{align*}
\end{definition}

\section{Different Perspective on Classical Information Cost}
	\label{sec:bas}

Before diving into the definition of quantum information cost 
and the properties of such a definition, we first present a different perspective
on the classical information cost that is both 
natural and more amenable to a quantum generalization.
Taking this perspective leads to an alternate proof of its operational interpretation
 as the amortized (distributional) communication cost.
The main difference from the standard definition is not so much in the formal rewriting of
this definition, which to some extent was already implicitly used in previous proofs \cite{BR11, Bra12a} and is simply an application of
the chain rule  and basic properties of mutual information. It is rather in the interpretation of every message transmission
as the simulation of a noisy channel (the generation of the message from the input, previous messages, and randomness)
with feedback to the sender and side information at the receiver, a
variant of the setting of the classical reverse Shannon theorem studied in the information theory literature \cite{BSST02, BDHSW09}.
Using a quantum analogue of the reverse Shannon theorem with side information at the receiver \cite{DY08, YD09}, 
this local description of information cost then circumvent usual difficulties in defining a quantum analogue of a transcript,
and leads to a generalization of information complexity with the desired properties.

We consider a $N$-message classical communication protocol $\Pi$,
along with a distribution $\mu$ over the inputs $ (X, Y)$,
and public randomness $R$. The protocol is defined by a sequence of conditional 
random variables $M_i$ taking value in the sample space
$\{0,1\}^*$, with some suitable constraints to enforce that protocols are well-defined.
For random variable $M, I$, we denote by $M|I$ the table of conditional transition probabilities
$p_{M|I} (M =m | I= i)$ that gives the probability to obtain output $M = m$ given input $I=i$.
Then, on input random variables $(X, Y)$, $\Pi$ is defined by
$M_1 | X R^A, M_2 | M_1^B Y R^B, M_3 | M_2^A M_1^A X R^A, \cdots,
 M_{N} | M_{N-1}^A 
\cdots 
M_1^A Y R^A$, with the superscripts denoting whose copy of a random variable we are considering.
We denote the transcript by $\Pi(x,y) = 
r \cdot m_1 \cdot m_2 \cdots m_{N}$, and the corresponding random 
variable by $\Pi(X,Y)$. The communication complexity of $\Pi$ is 
defined as $CC(\Pi) = \max | m_1 \cdot m_2 \cdots m_{N} |$, 
in which the maximum for the length is taken over all input pairs $(x, y) \in (X,Y)$, 
over all public randomness $r \in R$, and all conditional transcript 
$m_1 m_2 ... m_{N} \in M_1 M_2 \cdots M_{N} | (X,Y,R) = (x,y,r)$ 
(we only consider events with non-zero probability). Sometimes it is also
defined by taking the average length instead of the maximum; this does not affect the results here.

The standard definition of the information cost is then as the sum of two conditional mutual informations,
\begin{align}
IC_\mu (\Pi) = I(\Pi(X,Y); Y|X)+ I(\Pi(X,Y); X|Y),
\end{align}
 and the 
corresponding intuition for this quantity is that 
it represents the amount of information leaked by the transcript 
to Alice about Bob's input plus that leaked to Bob about Alice's input. 
This definition is shown in \cite{BR11} to lead to an equality between information complexity and 
 the distributional amortized communication 
complexity. We would like to show such a theorem 
about an analogous quantum information cost quantity.
However, many problems seem to arise when trying to generalize the 
above quantity to the quantum setting. 
Before proposing such a definition, 
we first give an alternate characterization of classical information cost, 
along with the corresponding operational intuition, which will be more 
amenable to a quantum generalization. Note that an interesting consequence of
 our perspective versus
those on sampling complexity that have been studied before is that
it enables us to use some previously proved tools from information theory, namely (tensor power source) classical reverse Shannon theorems.

We define an alternate information cost for classical communication 
protocols as 
\begin{align}
IC_\mu^\prime  = I(M_1^B; &X R^A | Y R^B) + 
I(M_2^A ; Y M_1^B R^B | X M_1^A R^A)  \\
&+ I(M_3^B ; X M_2^A M_1^A R^A | Y M_2^B M_1^B R^B) + \cdots  \\
&+ I(M_N^A; Y M_{N-1}^B \cdots M_1^B R^B | X M_{N-1}^A \cdots M_1^A R^A),
\end{align} 
in which we distinguish 
between Alice's and Bob's copy of the public 
randomness $R$ and messages $M_i$. Note that this is easily seen to be equivalent to the standard definition for
$IC_\mu$, by using the chain rule for mutual information along 
with the fact that $M_i^A, M_i^B$ and $R^A, R^B$
are just copies of one another. However, it is not 
so much in the formal statement that the rewriting is interesting, 
but in the operational interpretation. Indeed, the above characterization 
comes from viewing each conditional message $M|I$ in the protocol as 
a noisy channel in which the output $M$ is sent over a noiseless channel to a receiver
who has side information $S$ about the input $I$ to the channel, but for which
also a copy $M_F$ of the output is given as feedback to the sender.
The problem of simulating the sending of the output of
 a noisy channel with feedback has been studied in the literature under
 the name of classical reverse Shannon theorem \cite{BSST02, BDHSW09}, and when there is side information $S$
 at the receiver, $I(M;I|S)$ characterizes the amount of information that 
needs to be sent over the noiseless channel from sender to receiver. 
It is shown in \cite{LD09} that asymptotically, this task can be accomplished 
at a (unidirectional) classical communication rate of $I(M;I|S)$ when 
sufficient shared randomness is present. In \cite{BR11}, a correlated sampling 
protocol is used to perform a similar task in a one-shot setting, but such 
that the same communication efficiency is achieved on average, to first order. 
A caveat is that their protocol to do so is interactive, while the one
 in \cite{LD09} is not. This result yields a simulation protocol for amortized 
communication that asymptotically achieves communication at the information
cost of the protocol, while keeping the same round complexity,  and error parameter arbitrarily close
to the original one.
Another nice property of reverse Shannon theorems is that they also give
bounds on the amount of extra shared randomness required in the asymptotic limit for channel simulation.
We do not give the details of the proof here, it follows along the same line as 
the one  for the quantum case. We also do not discuss optimality 
here (see the quantum case, or \cite{BR11, Bra12a} for
such discussions).


\section{Quantum Information Cost and Complexity}
\label{sec:qic}

Finally, we are ready to define quantum information cost and complexity !
As we already said, the notion of information cost of a protocol does
not seem to easily extend to the quantum setting, mainly due to the 
fact that there is no direct analogue for a transcript in the quantum setting.
Also, the reversibility of 
quantum computation allows for protocols in which nothing remains at the 
end of the distributed computation except for the initial inputs and the 
output of the function evaluated on these inputs (up to some small error for approximate protocols). 
These issues brought 
Braverman \cite{Bra12a} to wonder what was the right quantum analogue of 
information cost, and whether there always existed protocols for computing 
binary functions that had quantum information cost bounded by a constant. 
Some attempts at trying to define a quantum analogue of information 
cost have appeared before \cite{JRS03, JN13}, and have proven useful for 
tackling particular problems. However, none of these seems to define the right 
notion in a context of amortized communication complexity.
In particular there is an implicit dependence on the round complexity 
for these notions, and they only provide a lower bound on the
 communication cost once divided by the number of rounds. Defining such a notion 
would hopefully lead to an interesting tool, in particular to obtain lower bounds and to tackle direct sum 
questions in quantum communication complexity. The main goal of our work 
is to present such a notion in a strong sense: we define a notion of 
quantum information complexity which is exactly equal to the amortized 
communication complexity. 

From the alternate definition of the information cost in the preceding section, 
we can more easily extend it to a notion of quantum information cost
for quantum protocols. A thing that might still cause problem is that
we cannot keep a copy of a channel input and output at the sender.
These issues have already been discussed in particular when discussing connections
between the fully quantum Slepian-Wolf theorem \cite{ADHW09} and the fully
quantum reverse Shannon theorem \cite{ BDHSW09, BCR11}, 
and when discussing the fully quantum generalization of channel simulation
with side information at the receiver \cite{LD09, DY08, YD09}. The correct quantum analogue of
this is that what stays at the sender is the coherent feedback of the environment output
of the noisy channel's isometric extension, and what is to be transmitted at the receiver is
the usual output of the channel. In our situation, there is also another system in play,
the side information already in the possession of the receiver before the transmission.
The correct problem to consider in this case is then quantum state redistribution \cite{LD09, DY08, YD09}.
In quantum state redistribution, there are $4$ systems of interest. 
At the outset of the protocol, the $A, C$ systems are in the possession of Alice, and would be
for us the coherent feedback of the noisy channel and the output to be transmitted, respectively,
while Bob holds the side information $B$, and the $ABC$ joint system is purified by a 
reference register $R$ that no party has access to. Thus, the only system changing hands 
is the $C$ subsystem that is to be transmitted from Alice to Bob. It is proved in \cite{DY08, YD09} that this
can be accomplished, in the limit of asymptotically many copies of this task, at a communication cost of
$\frac{1}{2} I(R; C | B)$ qubits per copy, along with an entanglement cost 
of $\frac{1}{2} I(C; A) - \frac{1}{2} I(C; B)$ ebits per copy
(with entanglement generated instead of consumed if this is negative). 
We state a precise formulation of this theorem.
\begin{theorem} (State redistribution \cite{DY08, YD09})
	\label{th:yd}
	For any $\epsilon, \delta > 0$, any state $\rho^{ABC}$ and any 
purification $\rho^{ABCR}$, any quantum communication rate 
$Q > \frac{1}{2} I(C; R | B)$ and entanglement consumption 
(or generation if negative) rate $E$ satisfying $Q + E > H(C | B)$, 
there is a large enough $n_0$ such that for all $n \geq n_0$, 
there exist an encoder $E \in C(A^{\otimes n} \otimes 
C^{\otimes n} \otimes T_A^{in}, A^{\otimes n} \otimes 
\hat{C} \otimes T_A^{out})$ and a decoder $D \in 
C(B^{\otimes n} \otimes \hat{C} \otimes T_B^{in}, B^{\otimes n} 
\otimes C^{\otimes n} \otimes T_B^{out})$ such that 
$\dim (T_A^{in}) = \dim (T_B^{in}) = 2^{\lceil \max (E, \delta) n \rceil}, 
\dim (T_A^{out}) = \dim (T_B^{out}) = 2^{- \lceil \min 
(0, E) n \rceil}, \dim (\hat{C}) = 2^{\lceil Q n \rceil}$, $\psi_{in}^{T_A^{in} T_B^{in}}, \psi_{out}^{T_A^{out} T_B^{out}}$ are maximally entangled states in $T_A^{in} \otimes T_B^{in}, T_A^{out} \otimes T_B^{out}$, respectively, and
\begin{align}
\| \Tra{T_A^{out} T_B^{out}} \circ D \circ E (\rho^{\otimes n} \otimes \psi_{in}) - \rho^{\otimes n} \|_{A^{\otimes n} B^{\otimes n} C^{\otimes n} R^{\otimes n}} \leq \epsilon, \\
\| \Tra{\neg T_A^{out} T_B^{out}} \circ D \circ E (\rho^{\otimes n} \otimes \psi_{in}) - \psi_{out} \|_{T_A^{out} T_B^{out}} \leq \epsilon.
\end{align}
\end{theorem}

Now, in analogy with our
rewriting of the classical information cost, we define the quantum information 
cost of a protocol, for a protocol $\Pi$ as defined in section \ref{sec:qucomm}, in the following way.

\begin{definition}
For a protocol $\Pi$ and an input state $\rho$, 
we define the \emph{quantum information cost} of $\Pi$ on input $\rho$ as
\begin{align*}
	QIC (\Pi, \rho) = \sum_{i>0, odd} \frac{1}{2} I(C_i; R | B_{i - 1}) + \sum_{i>0, even} \frac{1}{2} I(C_i; R | A_{i - 1}),
\end{align*}
in which we have labelled $B_0 = B_{in} \otimes T_B$.
\end{definition}

Note that even for protocols with non-zero communication, the quantum information cost on pure state input is zero, since
the purifying $R$ register is trivial in such a case. The corresponding notion of quantum information complexity of a channel is then:

\begin{definition}
For a bipartite channel $\N \in \C (A_{in} \otimes B_{in}, A_{out} \otimes B_{out})$, 
an input state $\rho~\in~\D(A_{in} \otimes B_{in})$ and an error parameter $\epsilon \in [0, 2]$, 
we define the $\epsilon$-error \emph{quantum information complexity} of $\N$ on input $\rho$ as
\begin{align*}
	QIC (\N, \rho, \epsilon) = \inf_{\Pi \in \T (\N, \rho, \epsilon)} QIC (\Pi, \rho).
\end{align*}
\end{definition}

We have the following operational interpretation for quantum information complexity as
the amortized quantum communication complexity.

\begin{theorem}
\label{th:qic}
For a bipartite channel $\N \in \C (A_{in} \otimes B_{in}, A_{out} \otimes B_{out})$, 
an input state $\rho~\in~\D(A_{in} \otimes B_{in})$ and an error parameter $\epsilon \in (0, 2]$, 
\begin{align*}
QIC (\N, \rho, \epsilon) = AQCC (\N, \rho, \epsilon).
\end{align*}
\end{theorem}

The proof of the theorem is given in section \ref{sec:qicam}. We can also obtain bounded round variants.

\begin{definition}
For a bipartite channel $\N \in \C (A_{in} \otimes B_{in}, A_{out} \otimes B_{out})$, 
an input state $\rho~\in~\D(A_{in} \otimes B_{in})$, an error parameter $\epsilon \in [0, 2]$ and a bound $M \in \mathbb{N}$ on the number of messages, 
we define the $M$-message, $\epsilon$-error \emph{quantum information complexity} of $\N$ on input $\rho$ as
\begin{align*}
	QIC^M (\N, \rho, \epsilon) = \inf_{\Pi \in \T^M (\N, \rho, \epsilon)} QIC (\Pi, \rho).
\end{align*}
\end{definition}

\begin{theorem}
\label{th:qicbndr}
For a bipartite channel $\N \in \C (A_{in} \otimes B_{in}, A_{out} \otimes B_{out})$, 
an input state $\rho~\in~\D(A_{in} \otimes B_{in})$, an error parameter $\epsilon \in (0, 2]$ and a bound $M \in \mathbb{N}$ on the number of messages, 
\begin{align*}
QIC^M (\N, \rho, \epsilon) = AQCC^M (\N, \rho, \epsilon).
\end{align*}
\end{theorem}

\section{Properties of the Definition}

\subsection{Quantum information lower bounds communication}

In this section, we make the important remark that in any protocol, the quantum information 
cost is non-negative and, more importantly, is a lower bound on the quantum communication cost.
This holds when considering the quantum information cost with respect to
any input state.
This follows from the fact for any quantum state, 
$0 \leq \frac{1}{2} I(C;R|B)  \leq \log \dim (C)$.
Applying this to all terms in the quantum information cost versus all 
terms in the quantum communication cost, we
get the result. A similar results holds for quantum information complexity 
versus quantum communication complexity, by taking infimum on both sides.

\begin{lemma}
\label{lem:qicvsqcc}
	For any protocol $\Pi$ and input state $\rho$, the following holds
\begin{align*}
0 \leq QIC (\Pi, \rho) \leq QCC (\Pi).
\end{align*}
\end{lemma}

\begin{corollary}
	For any channel $\N$, any input state $\rho$, any error parameter $\epsilon \in [0, 2]$ and any bound $M \in \mathbb{N}$ on the number of messages, the following holds
\begin{align*}
QIC (\N, \rho, \epsilon) & \leq QCC (\N, \rho, \epsilon), \\
QIC^M (\N, \rho, \epsilon) & \leq QCC^M (\N, \rho, \epsilon).
\end{align*}
\end{corollary}

\subsection{Quantum information upper bounds amortized communication}

We now prove some kind of converse result to the one in the previous section:
 the quantum information cost is an upper bound on
the amortized quantum communication cost.

\begin{lemma}
\label{lem:coding}
For any $M$-message protocol $\Pi$, any input state $\rho$ and any $\epsilon \in (0, 2], \delta > 0$,
there exists a large enough $n_0$ such that for any $n \geq n_0$, there exists
a protocol $\Pi_n \in \T^M (\Pi^{\otimes n}, \rho^{\otimes n}, \epsilon)$ satisfying
\begin{align*}
\frac{1}{n} QCC (\Pi_n) \leq QIC(\Pi, \rho) + \delta.
\end{align*}
\end{lemma}

\begin{proof}
Given any $M$-message protocol $\Pi$ and any state $\rho^{A_{in} B_{in} R}$, let 
\begin{align*}
\rho_1^{A_1 C_ 1 B_0 R} = U_1 (\rho \otimes \psi), \rho_
2^{A_1 C_2 B_2 R} = U_2 (\rho_1), \cdots,  \rho_M^{A_{M-1} 
C_M B_M R} = U_M (\rho_{M-1})
\end{align*}
 in which we label  
$B_0 = B_{in} \otimes T_B, B_M = B_{out} \otimes B^\prime$. 
Then, for any $\epsilon, 
\delta > 0$, take $Q_i = \frac{1}{2} 
I(C_i ; R| B_{i-1}) + \frac{\delta}{2 M}$ and $F_i = 
\frac{1}{2} \max (0, I(C_i ; A_{i}) - I(C_i; B_{i-1}) ) + \frac{\delta}{2 M}$ for 
$i$ odd
(we do not worry here about reusing the possibly generated entanglement, and simply discard it in the encoding and decoding
maps to be defined below), $Q_i = \frac{1}{2} I(C_i ; R | A_{i-1}) + \frac{\delta}{2 M}$ 
and $F_i = \frac{1}{2} \max (0,  I(C_i ; B_{i}) - I(C_i; A_{i-1}) ) + 
\frac{\delta}{2 M}$ for $i$ even,
 and for each $i$ let $n_0^i$ be
the corresponding $n_0$ for error $\frac{\epsilon}{M}$ in Theorem \ref{th:yd}, and take $n_0 = \max \{ n_0^i \}$.
Then for any $n \geq \max (n_0, \frac{2 M}{\delta})$, we have encoding and decoding maps $E_i, D_i$, along with corresponding entanglement
$\psi_i \in D(T_A^i \otimes T_B^i)$  and communication register $\hat{C}^i$ of size $\dim \hat{C}^i = 2^{\lceil Q_i n \rceil}$,
with each satisfying
\begin{align}
\label{eq:encdec1}
\| D_i \circ E_i (\rho_i^{\otimes n} \otimes \psi_i) - \rho_i^{\otimes n} 
\|_{A_{i }^{\otimes n} C_i^{\otimes n} B_{i-1 }^{\otimes n} R^{\otimes n}} \leq \frac{\epsilon}{M}
\end{align}
 for odd $i$, or 
\begin{align}
\label{eq:encdec2}
\| D_i \circ E_i (\rho_i^{\otimes n} \otimes \psi_i) - \rho_i^{\otimes n} 
\|_{A_{i-1}^{\otimes n} C_i^{\otimes n} B_{i}^{\otimes n} R^{\otimes n}} \leq \frac{\epsilon}{M}
\end{align} 
for even $i$.
Let $\hat{E}_i, \hat{D}_i$ be unitary extensions of $E_i, D_i$, respectively, requiring ancillary states
$\sigma_i^E \in \D (E_i^{in}), \sigma_i^D \in \D (D_i^{in})$.
We define the following protocol $\Pi_n$ starting from the protocol $\Pi$.

\hrulefill

Protocol $\Pi_n$ on input $\sigma$: \\
-Take entangled state $\hat{\psi} = 
\psi^{\otimes n} \otimes \psi_1 \otimes \sigma_1^E \otimes \sigma_1^D \otimes \cdots \otimes \psi_M \otimes \sigma_M^E \otimes \sigma_M^D$. \\
-Take unitaries $\hat{U}_1 = \hat{E}_i \circ U_1^{\otimes n}, 
\hat{U}_2 = \hat{E}_2 \circ U_2^{\otimes n} \circ \hat{D}_1, \cdots, 
\hat{U_M} = \hat{E}_M \circ U_M^{\otimes n} \circ \hat{D}_{M-1},
\hat{U}_{M+1} =  U_{M+1}^{\otimes n} \circ \hat{D}_{M}$  \\
-Take as output the $A_{out}^{\otimes n},  B_{out}^{\otimes n}$ registers.

\hrulefill

Note that the communication cost of $\Pi_n$ satisfies
\begin{align*}
QCC (\Pi_n) & = \sum_i \log \dim (\hat{C}^i) \\
	&= \sum_i \lceil Q_i n \rceil \\
	& \leq  n(\sum_{i>0, odd} \frac{1}{2} I(C_i; R | B_{i - 1}) + \sum_{i>0, even} \frac{1}{2} I(C_i; R | A_{i - 1}) + \frac{ M \delta}{2 M} + \frac{M}{n}) \\
	& \leq n (QIC (\Pi, \rho) + \delta).
\end{align*}
This is also a $M$-message protocol, so is left to bound the error on input $\sigma = \rho^{\otimes n}$ to make sure that $\Pi_n \in \T (\Pi^{\otimes n}, \rho^{\otimes n}, \epsilon)$. We have
\begin{align*}
\| \Pi_n (\rho^{\otimes n}) - \Pi^{\otimes n} (\rho^{\otimes n}) \| 
	= \| \Tra{\neg A_{out}^{\otimes n}  B_{out}^{\otimes n}} & U_{M+1}^{\otimes n} \hat{D}_M \hat{E}_M U_M^{\otimes n} \hat{D}_{M-1} \cdots \hat{E}_1 U_1^{\otimes n} (\rho^{\otimes n} \otimes \hat{\psi}) \\
	& - \Tra{\neg A_{out}^{\otimes n} B_{out}^{\otimes n}}  U_{M+1}^{\otimes n}  U_M^{\otimes n}  \cdots  U_1^{\otimes n} (\rho^{\otimes n} \otimes \psi^{\otimes n}) \| \\
	= \| \Tra{\neg A_{out}^{\otimes n} B_{out}^{\otimes n}} & U_{M+1}^{\otimes n} D_M E_M U_M^{\otimes n} D_{M-1} \cdots E_1 U_1^{\otimes n} (\rho^{\otimes n} \otimes \psi^{\otimes n} \otimes \psi_1 \otimes \cdots \otimes \psi_M) \\
	& - \Tra{\neg A_{out}^{\otimes n} B_{out}^{\otimes n}}  U_{M+1}^{\otimes n}  U_M^{\otimes n}  \cdots  U_1^{\otimes n} (\rho^{\otimes n} \otimes \psi^{\otimes n}) \| \\
	 \leq \| \Tra{(A^\prime)^{\otimes n}  (B^\prime)^{\otimes n}} & U_{M+1}^{\otimes n} D_M \cdots E_2 U_2^{\otimes n} D_1 E_1  (\rho_1^{\otimes n} \otimes \psi_1 \otimes \cdots \otimes \psi_M) \\
	& - \Tra{ (A^\prime)^{\otimes n} (B^\prime)^{\otimes n}}  U_{M+1}^{\otimes n} D_M \cdots E_2 U_2^{\otimes n} (\rho_1^{\otimes n} \otimes \psi_2 \otimes \cdots \otimes \psi_M) \| \\
	 + \| \Tra{(A^\prime)^{\otimes n}  (B^\prime)^{\otimes n}} & U_{M+1}^{\otimes n} D_M E_M U_M^{\otimes n} D_{M-1} \cdots U_3^{\otimes n} D_2 E_2 (\rho_2^{\otimes n}  \otimes \psi_2 \otimes \cdots \otimes \psi_M) \\
	  - \Tra{(A^\prime)^{\otimes n} (B^\prime)^{\otimes n}} & U_{M+1}^{\otimes n} D_M E_M U_M^{\otimes n} D_{M-1} \cdots E_3 U_3^{\otimes n} (\rho_2^{\otimes n} \otimes \psi_3 \otimes \cdots \otimes \psi_M) \| \\
	& + \cdots \\
	 + \| \Tra{(A^\prime)^{\otimes n} (B^\prime)^{\otimes n}} & U_{M+1}^{\otimes n} D_M E_M U_M^{\otimes n} D_{M-1} E_{M-1}  (\rho_{M-1}^{\otimes n} \otimes \psi_{M-1} \otimes \psi_M) \\
	& - \Tra{(A^\prime)^{\otimes n} (B^\prime)^{\otimes n}} U_{M+1}^{\otimes n} D_M E_M U_M^{\otimes n} (\rho_{M-1}^{\otimes n} \otimes \psi_M ) \| \\
	 + \| \Tra{(A^\prime)^{\otimes n} (B^\prime)^{\otimes n}} & U_{M+1}^{\otimes n} D_M E_M  (\rho_M^{\otimes n} \otimes \psi_M ) \\
	& - \Tra{(A^\prime)^{\otimes n} (B^\prime)^{\otimes n}} U_{M+1}^{\otimes n}  (\rho_M^{\otimes n} ) \| \\
	 \leq \|  D_1 E_1  (\rho_1^{\otimes n} \otimes & \psi_1 \otimes \psi_2 \otimes \cdots \otimes \psi_M) -  (\rho_1^{\otimes n} \otimes \psi_2 \otimes \cdots \otimes \psi_M) \| \\
	 + \|  D_2 E_2 (\rho_2^{\otimes n}  \otimes & \psi_2 \otimes \psi_3 \otimes \cdots \otimes \psi_M) -  (\rho_2^{\otimes n} \otimes \psi_3 \otimes \cdots \otimes \psi_M) \| \\
	& + \cdots \\
	 + \|  D_{M-1} E_{M-1} & (\rho_{M-1}^{\otimes n}  \otimes  \psi_{M-1} \otimes \psi_M) -  (\rho_{M-1}^{\otimes n} \otimes \psi_M ) \| \\
	 + \|  D_M E_M  (\rho_M^{\otimes n} & \otimes \psi_M ) -   (\rho_M^{\otimes n} ) \| \\
	 \leq M \frac{\epsilon}{M} & \\
	 = \epsilon. &
\end{align*}
The first equality is by definition, the second one by tracing the registers $E_i^{out}, D_i^{out}$ from the unitary extensions to the encoders and decoders in the first term, the first inequality is by the triangle inequality and by definition of the $\rho_i$'s, the second inequality is due to the monotonicity of trace distance under noisy channels, and the next is by (\ref{eq:encdec1}) and (\ref{eq:encdec2}), along with the fact that appending uncorrelated systems does not change the trace distance.

We also note that on top of the entanglement $\psi$ used by each compressed copy of the protocol, the asymptotic entanglement consumption (or generation) rate is bounded by the initial quantum communication cost, which follows from $- \log \dim (C) \leq \frac{1}{2} ( I(C;A) - I(C;B) ) \leq \log \dim (C)$.

\end{proof}

\subsection{Additivity}

We now show that quantum information complexity satisfy an additivity 
property. This is used in the converse part of the proof of Theorem \ref{th:qic}. 
Let us set some notation first. We say that a triple $(\N, \rho, \epsilon)$ 
is a quantum task, corresponding to the simulation of channel $\N \in C (A_{in} \otimes B_{in}, A_{out} \otimes B_{out})$ 
on input $\rho \in \D (A_{in} \otimes B_{in})$ with error $\epsilon \in [0, 2]$.
We define a product quantum task recursively, and with the following notation:
a quantum task is a product quantum task, and if $T_1 = (\N_1, \rho_1, \epsilon_1) \otimes \cdots \otimes (\N_i, \rho_i, \epsilon_i), T_2 = (\N_{i+1}, \rho_{i+1}, \epsilon_{i+1}) \otimes \cdots \otimes (\N_{n}, \rho_{n}, \epsilon_{n})$ are two product quantum tasks, then $T_1 \otimes T_2 = \bigotimes_{i \in [n]} (\N_i, \rho_i, \epsilon_i)$ is also a product quantum task. 
We say that a protocol $\Pi_n$, with input space $A_{in}^1 \otimes B_{in}^1 \otimes \cdots \otimes A_{in}^n \otimes B_{in}^n$ and output space $A_{out}^1 \otimes B_{out}^1 \otimes \cdots \otimes A_{out}^n \otimes B_{out}^n$,
succeeds at the product quantum task 
$\bigotimes_i (\N_i, \rho_i, \epsilon_i)$
if it succeeds, for each $i$, at simulating 
channel $\N_i$ on input $\rho_i$ with error $\epsilon_i$, 
and denote by
$\T_{\otimes} ( \bigotimes_i (\N_i, \rho_i, \epsilon_i))$ 
the set of all protocols achieving this. 
Once again, if we restrict this set to $M$-message protocols, we write $\T_{\otimes}^M ( \bigotimes_i (\N_i, \rho_i, \epsilon_i))$.
We then define the quantum information 
complexity of the product quantum task $\bigotimes_i (\N_i, \rho_i, \epsilon_i)$ as 
\begin{align}
QIC_\otimes (\bigotimes_i (\N_i, \rho_i, \epsilon_i) ) = \inf_{\Pi_n \in \T_\otimes  (\bigotimes_i (\N_i, \rho_i, \epsilon_i))} QIC (\Pi_n, \rho_1 \otimes \cdots \otimes \rho_n ).
\end{align}
For the bounded round variant, we have
\begin{align}
QIC_\otimes^M (\bigotimes_i (\N_i, \rho_i, \epsilon_i) ) = \inf_{\Pi_n \in \T_\otimes^M  (\bigotimes_i (\N_i, \rho_i, \epsilon_i))} QIC (\Pi_n, \rho_1 \otimes \cdots \otimes \rho_n ).
\end{align}
We first prove the following two technical lemmata that will lead to the additivity result.

\begin{lemma}
\label{lem:add1}
	For any two protocols $\Pi^1, \Pi^2$ with $M_1, M_2$ messages, respectively,
there exists a $M$-message protocol $\Pi_2$, satisfying $\Pi_2 = \Pi^1 \otimes \Pi^2, M = \max (M_1, M_2)$, such that the following holds
for any corresponding input states $\rho^1, \rho^2$:
\begin{align*}
QIC (\Pi_2, \rho^1 \otimes \rho^2) = QIC(\Pi^1, \rho^1) + QIC(\Pi^2, \rho^2).
\end{align*}
\end{lemma}

\begin{proof}
Given protocols $\Pi^1$ and $\Pi^2$, we assume without loss of generality that $M_1 \geq M_2$, and we define the protocol $\Pi_2$ in the following way.

\hrulefill

Protocol $\Pi_2$ on input $\sigma$: \\
-Run protocols $\Pi^1, \Pi^2$ in parallel for $M_2$ messages, on corresponding input registers $A_{in}^1, B_{in}^1, A_{in}^2, B_{in}^2$ until $\Pi^2$ has finished. \\
-Finish running protocol $\Pi^1$. \\
-Take as output the output registers $A_{out}^1, B_{out}^2, A_{out}^2, B_{out}^2$ of both $\Pi^1$ and $\Pi^2$.

\hrulefill

It is clear that the channel that $\Pi_2$ implements is $\Pi_2 = \Pi^1 \otimes \Pi^2$,
and the number of messages satisfies $M = \max (M_1, M_2)$,
so is left to analyse its quantum information cost on input $\sigma = \rho_1 \otimes \rho_2$.
The first thing to notice is that we can find a purification of 
$\rho_1 \otimes \rho_2$ that is also in product form, i.e.~there 
exist a purification with the purifying system 
$R = R^1 \otimes R^2$ and such that $(\rho_1 \otimes 
\rho_2)^{A_{in}^1 B_{in}^1 A_{in}^2 B_{in}^2 R} = 
\rho_1^{A_{in}^1 B_{in}^1 R^1} \otimes \rho_2^{A_{in}^2 B_{in}^2 R^2}$. 
Also note that throughout the protocol, 
due to the structure of $\Pi_2$ and the fact that the input $\rho_1 \otimes \rho_2$ is in product form, 
any registers corresponding to 
$\Pi^1$ stays in product form with any register corresponding 
to $\Pi^2$. We label, for
 $i \in \{1, 2 \}, A_0^i = A_{in}^i \otimes T_A^i, B_0^i = B_{in}^i \otimes
T_B^i, A_{M_i}^i = A_{out}^i \otimes (A^{\prime})^i, B_{M_i}^i = B_{out}^i \otimes (B^{\prime})^i$. Then
\begin{align*}
		 QIC (\Pi_2, \rho_1 \otimes \rho_2)  
		= I(C_1^1 C_1^2 ; R^1 R^2 | B_0^1 B_0^2) &+ I(C_2^1 C_2^2 ; R^1 R^2 | A_1^1 A_1^2) \\
		&+ \cdots +I(C_{M_2}^1 C_{M_2}^2 ; R^1 R^2 | A_{M_2-1}^1 A_{M_2-1}^2) \\
		+ I(C_{M_2 + 1}^1  ; R^1 R^2 | & B_{M_2}^1  B_{M_2}^2) + \cdots + I(C_{M_1}^1  ; R^1 R^2 | A_{M_1-1}^1 A_{M_2}^2) \\
		= I(C_1^1 ; R^1 | B_0^1 ) + I ( & C_2^1  ; R^1 | A_1^1 ) + \cdots +I(C_{M_1}^1 ; R^1 | A_{M_1-1}^1 ) \\
		+ I(C_1^2 ; R^2 |  B_0^2) + & I(C_2^2 ; R^2 |  A_1^2) + \cdots \\		
		= QIC (\Pi^1, \rho_1) +  Q&IC  (\Pi^2, \rho_2).
\end{align*}
The first equality is by definition of quantum information cost of $\Pi_2$, and due to its parallel
 structure, the second equality is because registers of $\Pi^1, \Pi^2$ are in product form, 
and then the last equality follows from definition and the structure of $\Pi_2$. 
\end{proof}

\begin{lemma}
\label{lem:add2}
	For any $M$-message protocol $\Pi_2$ and any input states $\rho^1 \in \D (A_{in}^1 \otimes B_{in}^1), \rho_2 \in \D ( A_{in}^2 \otimes B_{in}^2)$,
there exist $M$-message protocols $\Pi^1, \Pi^2$ satisfying $\Pi^1 (\cdot) = \Tra{A_{out}^2 B_{out}^2} \circ \Pi_2 (\cdot \otimes \rho^2), \Pi^2 (\cdot) = \Tra{A_{out}^1 B_{out}^1} \circ \Pi_2 (\rho^1 \otimes \cdot)$, and the following holds:
\begin{align*}
QIC(\Pi^1, \rho^1) + QIC(\Pi^2, \rho^2) = QIC (\Pi_2, \rho^1 \otimes \rho^2).
\end{align*}
\end{lemma}

\begin{proof}
Given $\Pi_2$, we define the 
protocols $\Pi^1, \Pi^2$ in the following way.

\hrulefill

Protocol $\Pi^1$ on input $\sigma^1$: \\
-Let $(\rho^2)^{A_{in}^2 B_{in}^2 R^2}$ be a purification of $\rho^2$, 
and $\psi^{T_A T_B}$ be the entangled state used in the $\Pi_2$ protocol. 
The entangled state for the protocol will be $\rho^2 \otimes \psi$, 
with the $A_{in}^2, R^2, T_A $ registers given to Alice, 
and the $B_{in}^2, T_B$ registers given to Bob. \\
-Using the $\rho^2$ state given as pre-shared entanglement to simulate 
the other input, run protocol $\Pi_2$ on input $\sigma_1 \otimes \rho_2$. \\
-Take as output the $A_{out}^1 B_{out}^1$ output registers.

\hrulefill

Protocol $\Pi^2$ on input $\sigma^2$: \\
-Let $(\rho^1)^{A_{in}^1 B_{in}^1 R^1}$ be a purification of $\rho^1$, 
and $\psi^{T_A T_B}$ be the entangled state used in the $\Pi_2$ protocol. 
The entangled state for the protocol will be $\rho^1 \otimes \psi$, 
with the $A_{in}^1, T_A $ registers given to Alice, 
and the $B_{in}^2, R^1, T_B$ registers given to Bob. \\
-Using the $\rho^1$ state given as pre-shared entanglement 
to simulate the other input, run protocol $\Pi_2$ on input $\rho^1 \otimes \sigma^2$. \\
-Take as output the $A_{out}^2 B_{out}^2$ output registers.

\hrulefill

It is clear that for any $\sigma^1 \in \D (A_{in}^1 \otimes B_{in}^1), \sigma^2 \in \D (A_{in}^2 \otimes B_{in}^2)$, $\Pi^1, \Pi^2$ implement $\Pi^1 (\sigma^1) = \Tra{A_{out}^2 B_{out}^2} 
\circ \Pi_2 (\sigma^1 \otimes \rho^2), \Pi^2 (\sigma^2)  = \Tra{A_{out}^1 B_{out}^1} \circ \Pi_2 (\rho^1 \otimes \sigma^2)$, respectively.
Also, $\Pi^1, \Pi^2$ are $M$-message protocols, 
 so is left to analyse their quantum information costs on input $\rho^1, \rho^2$, respectively. 
We reuse the notation from the previous lemma. By definition and the structure of the protocols, 
\begin{align*}
 QIC (\Pi^1, \rho^1) 
		&= I(C_1 ; R^1 | B_0 ) + I(C_2 ; R^1 | A_1 R^2 ) + \cdots, \\
QIC (\Pi^2, \rho^2) 
		&= I(C_1 ; R^2 | B_0 R^1 ) + I(C_2 ; R^2 | A_1 ) + \cdots,
\end{align*}
and we get by rearranging terms
\begin{align*}
	QIC (\Pi^1, \rho^1) + QIC (\Pi^2, \rho^2) 
		& = I(C_1 ; R^1 | B_0 ) + I(C_1 ; R^2 | B_0 R^1 ) \\
		&+ I(C_2 ; R^1 | A_1 R^2 ) + I(C_2 ; R^2 | A_1 ) + \cdots \\
		&= I(C_1 ; R^1 R^2 | B_0) + I(C_2 ; R^1 R^2 | A_1 ) + \cdots \\
		&= QIC (\Pi_2, \rho_1 \otimes \rho_2)
\end{align*}
where we have used the structure of the protocols along with the chain rule for mutual information on pairs of terms for 
the second equality, and the last equality follow from definition and the structure of the protocols. 

\end{proof}

We then get as a corollary the following additivity result:

\begin{corollary}
\label{cor:add3}
	For any two product quantum tasks $T_1, T_2$ and any bound $M \in \mathbb{N}$ on the number of messages,
\begin{align*}
	QIC_\otimes (T_1 \otimes T_2) & = 
QIC_\otimes (T_1) + QIC_\otimes (T_2), \\
	QIC_\otimes^M (T_1 \otimes T_2) & = 
QIC_\otimes^M (T_1) + QIC_\otimes^M (T_2).
\end{align*}
\end{corollary}

\begin{proof}
We consider the two product tasks $T_1 = \bigotimes_i (\N_i, \rho_i, \epsilon_i), T_2 = \bigotimes_j (\M_j, \sigma_j, \delta_j)$.
We first prove the $\leq$ direction.
Let $\Pi^1$ and $\Pi^2$ be protocols succeeding at the corresponding tasks $T_1, T_2$,
and achieving, for an arbitrary 
small $\epsilon^\prime > 0$, $QIC (\Pi^1, \rho_1 \otimes \cdots \otimes \rho_n) 
\leq QIC_\otimes (T_1) + \epsilon^\prime, QIC (\Pi^2, \sigma_1 \otimes \cdots \otimes \sigma_m) 
\leq QIC_\otimes (T_2) + \epsilon^\prime$, respectively. Taking the corresponding protocol $\Pi_2$ from Lemma \ref{lem:add1},
it clearly succeeds at the product task $T_1 \otimes T_2$,
and we get 
\begin{align*}
		QIC_\otimes (T_1 \otimes T_2)
		&\leq QIC (\Pi_2, \rho_1 \otimes \cdots \otimes \rho_n \otimes \sigma_1 \otimes \cdots \otimes \sigma_m)  \\
		&= QIC (\Pi^1, \rho_1 \otimes \cdots \otimes \rho_n) + QIC (\Pi^2, \sigma_1 \otimes \cdots \otimes \sigma_m) \\
		&\leq QIC_\otimes (T_1) + QIC_\otimes (T_2) + 2 \epsilon^\prime.
\end{align*}

Now for the $\geq$ direction, let $\Pi_2$ be a protocol succeeding at the product 
task and achieving $QIC (\Pi_2, \rho_1 \otimes \cdots \otimes \rho_n \otimes \sigma_1 \otimes \cdots \otimes \sigma_m) \leq 
QIC_\otimes (T_1 \otimes T_2) 
+ \epsilon^\prime$ for an arbitrary small $\epsilon^\prime > 0$. Taking the corresponding protocols $\Pi^1, \Pi^2$
from Lemma \ref{lem:add2} for tasks $T_1, T_2$, they clearly succeed at their respective task, and we get
\begin{align*}
	QIC_\otimes (T_1) &+ QIC_\otimes (T_2) \\
		& \leq  QIC (\Pi^1, \rho_1 \otimes \cdots \otimes \rho_n) + QIC (\Pi^2, \sigma_1 \otimes \cdots \otimes \sigma_m) \\
		&= QIC (\Pi_2, \rho_1 \otimes \cdots \otimes \rho_n \otimes \sigma_1 \otimes \cdots \otimes \sigma_m) \\
		&\leq  QIC_\otimes (T_1 \otimes T_2) + \epsilon^\prime.
\end{align*}
Keeping tracks of rounds, we also get the bounded round result.
\end{proof}

If we now consider $n$-fold product quantum task $(\N, \rho, \epsilon)^{\otimes n}$ and compare to the notation introduced when discussing amortized communication, we have $\T_n (\N^{\otimes n}, \rho^{\otimes n}, \epsilon) = \T_\otimes ((\N, \rho, \epsilon)^{\otimes n})$. Correspondingly, we define
\begin{align}
QIC_n (\N^{\otimes n}, \rho_n^{\otimes n}, \epsilon) = QIC_\otimes ((\N, \rho, \epsilon)^{\otimes n}),
\end{align}
and the bounded round variant
\begin{align}
QIC_n^M (\N^{\otimes n}, \rho_n^{\otimes n}, \epsilon) = QIC_\otimes^M ((\N, \rho, \epsilon)^{\otimes n}).
\end{align}

Then we get by induction on the previous corollary the additivity result used in the proof of Theorem \ref{th:qic}.

\begin{corollary}
\label{cor:add}
	For any quantum task $(\N, \rho, \epsilon)$ and bound $M \in \mathbb{N}$ on the number of messages,
\begin{align*}
	n QIC (\N, \rho, \epsilon) & = QIC_n (\N^{\otimes n}, \rho^{\otimes n}, \epsilon), \\
	n QIC^M (\N, \rho, \epsilon) & = QIC_n^M (\N^{\otimes n}, \rho^{\otimes n}, \epsilon).
\end{align*}
\end{corollary}

By definition and by using Lemma \ref{lem:qicvsqcc}, we also get the following.

\begin{corollary}
\label{cor:qicnqccn}
	For any quantum task $(\N, \rho, \epsilon)$ and bound $M \in \mathbb{N}$ on the number of messages,
\begin{align*}
	QIC_n (\N^{\otimes n}, \rho^{\otimes n}, \epsilon) 
& \leq QCC_n (\N^{\otimes n}, \rho^{\otimes n}, \epsilon), \\
	QIC_n^M (\N^{\otimes n}, \rho^{\otimes n}, \epsilon) 
& \leq QCC_n^M (\N^{\otimes n}, \rho^{\otimes n}, \epsilon).
\end{align*}
\end{corollary}

\subsection{Convexity, Concavity and Continuity}

We show that quantum information complexity is jointly convex in the channel and the error parameter.
We also state the corollary that it is continuous in the error parameter, a fact used in the direct coding part of 
Theorem \ref{th:qic}. 
Finally, we prove that the quantum information cost is concave in the input state.

We start with convexity.

\begin{lemma}
\label{lem:conv}
For any $p \in [0, 1]$, any two protocols $\Pi^1, \Pi^2$ with $M_1, M_2$ messages, respectively, there exists a $M$-message protocol $\Pi$ satisfying
$\Pi = p\Pi^1 +(1-p)\Pi^2, M = \max (M_1, M_2)$, such that the following holds for and any state $\rho$:
\begin{align*}
 QIC(\Pi, \rho) = p QIC(\Pi^1, \rho) + (1-p) QIC(\Pi^2, \rho).
\end{align*}
\end{lemma}

\begin{proof}
Given $\Pi^1, \Pi^2$, we assume without loss of generality that $M_1 \geq M_2$, and we define $\Pi$ in the following way:

\hrulefill

Protocol $\Pi$ on input $\rho$: \\
-The entangled state $\psi$ contains many parts: it contains both entangled states $\psi_1, \psi_2$ 
for the corresponding protocols $\Pi^1, \Pi^2$, it contains selector registers in 
state $\ket{\sigma_p} = \sqrt{p}\ket{1}^{S_A} \ket{1}^{S_B} + 
\sqrt{1-p} \ket{2}^{S_A} \ket{2}^{S_B}$, and it contains padding 
pure states to feed as input to the protocol that is not selected, held in registers $D_A, D_B$. \\
-Coherently control what to input into the two protocols: on control 
set to $1$, input state $\rho$ into protocol $\Pi^1$ and the padding 
pure state into protocol $\Pi^2$ , and vice-versa on control set to $2$. \\
-Run protocols $\Pi^1, \Pi^2$ in parallel for $M_2$ messages on given input until $\Pi^2$ has finished. \\
-Finish running protocol $\Pi^1$. \\
-Coherently control what to output: on control set to $1$, take as output 
the $A_{out}, B_{out}$ registers of protocol $\Pi^1$, and on control set to $2$, take those of protocol $\Pi^2$.

\hrulefill

Note that by the structure of the above protocol and because the selector registers are traced out at the end,
the output is of the form $\Pi (\rho) = p \Pi^1 (\rho) + (1-p) \Pi^2 (\rho)$, and $\Pi$ is a $M$-message protocol.
We must now verify that the quantum information cost satisfies the stated property.
First note that if Alice's registers are traced out, then Bob's selector register is effectively a classical register,
and similarly for Alice's selector register if Bob's registers are traced out.
Also note that throughout the protocol,
due to the structure of $\Pi$, conditional on some classical state of the selector register,
the reference register $R$ can only be correlated with
registers in the corresponding protocol, and
$D_B$ either contains a padding pure state or the input, with the later only when a padding pure state
was input into $\Pi$. Also, still conditional on some classical state of the selector register, any register
corresponding to $\Pi^1$ is in product form with any register corresponding to $\Pi^2$. Then 
\begin{align*}
		QIC (\Pi, \rho) 
		&= I(C_1^1 C_1^2 ; R | B_{in} D_B T_B^1 T_B^2 S_B) + I(C_2^1 C_2^2 ; R | A_1^1 A_1^2 S_A) \\
		& + \cdots +I(C_{M_2}^1 C_{M_2}^2 ; R | A_{M_2-1}^1 A_{M_2-1}^2 S_A) \\
		&+ I(C_{M_2 + 1}^1 ; R | B_{M_2}^1 B_{out}^2 (B^\prime)^2 S_B) + \cdots + I(C_{M_1}^1 ; R | A_{M_1-1}^1 A_{out}^2 (A^\prime)^2  S_A) \\
		&= p ( I(C_1^1 C_1^2 ; R | B_{in} D_B T_B^1 T_B^2 (S_B = 1)) + I(C_2^1 C_2^2 ; R | A_1^1 A_1^2 (S_A = 1)) \\
		& + \cdots +I(C_{M_2}^1 C_{M_2}^2 ; R | A_{M_2-1}^1 A_{M_2-1}^2 (S_A = 1)) \\
		&+ I(C_{M_2 + 1}^1 ; R | B_{M_2}^1 B_{out}^2 (B^\prime)^2 (S_B = 1)) + \cdots + I(C_{M_1}^1 ; R | A_{M_1-1}^1 A_{out}^2 (A^\prime)^2  (S_A = 1)) ) \\
		&+ (1-p) ( I(C_1^1 C_1^2 ; R | B_{in} D_B T_B^1 T_B^2 (S_B = 2)) + I(C_2^1 C_2^2 ; R | A_1^1 A_1^2 (S_A = 2)) \\
		& + \cdots +I(C_{M_2}^1 C_{M_2}^2 ; R | A_{M_2-1}^1 A_{M_2-1}^2 (S_A = 2)) \\
		&+ I(C_{M_2 + 1}^1 ; R | B_{M_2}^1 B_{out}^2 (B^\prime)^2 (S_B = 2)) + \cdots + I(C_{M_1}^1 ; R | A_{M_1-1}^1 A_{out}^2 (A^\prime)^2  (S_A = 2)) ) \\
		&= p ( I(C_1^1 ; R | B_{in}^1 T_B^1 (S_B = 1) ) + I(C_2^1 ; R | A_1^1 (S_A = 1)) + 
			\cdots +I(C_{M_1}^1 ; R | A_{M_1-1}^1 (S_A = 1))) \\
		&+ (1-p) (I(C_1^2 ; R | B_{in}^2 T_B^2 (S_B = 2)) + I(C_2^2 ; R |  A_1^2 (S_A = 2)) + \cdots ) \\
		&= p QIC (\Pi^1, \rho) +  (1-p) QIC (\Pi^2, \rho).
\end{align*}
The first equality is by definition of
 quantum information cost of $\Pi$, and due to its parallel
 structure, the second equality uses the above remark about the selector register of one party being
classical when the registers of the other party are traced out, along with a convex rewriting of conditional mutual information,
the third equality uses the above remark about the product structure of $R$ and the registers corresponding to $\Pi^1, \Pi^2$, respectively, depending on 
the classical state of the selector register, and the last equality is due to the fact that conditional on some classical state of the selector register,
the state in the registers considered is the same as the one in the corresponding protocol.
\end{proof}

We get the convexity in the channel and error parameters as a corollary.

\begin{corollary}
\label{lem:conv}
For any $p \in [0, 1]$, define $\N = p \N_1 + (1-p) \N_2$ for any two channels $\N_1, \N_2, \epsilon = p\epsilon_1 + (1-p) \epsilon_2$ for any two error parameters $\epsilon_1, \epsilon_2 \in [0, 2]$, for any bound $M = \max (M_1, M_2 ), M_1, M_2 \in \mathbb{N}$ on the number of messages and for any input state $\rho$, the following holds:
\begin{align*}
 QIC(\N, \rho, \epsilon) & \leq p QIC(\N_1, \rho, \epsilon_1) + (1-p) QIC(\N_2, \rho, \epsilon_2), \\
 QIC^M (\N, \rho, \epsilon) & \leq p QIC^{M_1} (\N_1, \rho, \epsilon_1) + (1-p) QIC^{M_2} (\N_2, \rho, \epsilon_2).
\end{align*}
\end{corollary}

\begin{proof}
Let $\Pi^1$ and $\Pi^2$ be protocols satisfying,
 for $i \in \{ 1, 2\}, \Pi^i \in \T (\N_i, \rho, \epsilon_i), QIC (\Pi^i, \rho) 
\leq QIC (\N_i, \rho, \epsilon_i) + \delta$ for an arbitrary 
small $\delta > 0$, and take the corresponding protocol $\Pi$ of Lemma \ref{lem:conv}.
We first verify that protocol $\Pi$ successfully accomplish its task. The result follows from the triangle inequality for the trace distance:
\begin{align*}
\| \Pi (\rho) - \N (\rho) \|_{A_{out} B_{out} R} &= \| p \Pi^1 (\rho) + (1-p) \Pi^2 (\rho) 
- (p \N_1 + (1-p) \N_2) (\rho) \|_{A_{out} B_{out} R} \\
&\leq \| p \Pi^1 (\rho)  -  p\N_1 (\rho) \|_{A_{out} B_{out} R} \\
& + \|  (1-p) \Pi^2 (\rho) - (1-p) \N_2 (\rho) \|_{A_{out} B_{out} R} \\
&\leq p \epsilon_1 + (1-p) \epsilon_2 \\
&= \epsilon.
\end{align*}
We must now verify that the quantum information cost satisfies the convexity property:
\begin{align*}
 QIC(\N, \rho, \epsilon) & \leq
		QIC (\Pi, \rho)  \\
		&= p QIC (\Pi^1, \rho) +  (1-p) QIC (\Pi^2, \rho) \\
		&\leq p QIC(\N_1, \rho, \epsilon_1) + (1-p) QIC(\N_2, \rho, \epsilon_2) + 2 \delta.
\end{align*}
Keeping track of rounds, we get the bounded round result.
\end{proof}

We get as a corollary that quantum information complexity is continuous
 in its error parameter. Note that quantum information complexity is decreasing in the error parameter.

\begin{corollary} 
\label{cor:cty}
For any $\N, \rho, \epsilon \in (0, 2], \delta > 0, M \in \mathbb{N}$, there exists an $0 < \epsilon^\prime < \epsilon$ such that the following holds:
\begin{align*}
0 \leq  QIC(\N, \rho, \epsilon- \epsilon^\prime) & - QIC(\N, \rho, \epsilon)   \leq \delta, \\
0 \leq  QIC^M (\N, \rho, \epsilon- \epsilon^\prime) & - QIC^M (\N, \rho, \epsilon)   \leq \delta.
\end{align*}
\end{corollary}

We now show that quantum information is concave in its input state parameter.

\begin{lemma} 
For any $p \in [0, 1]$, define $\rho = p\rho_1 + (1-p) \rho_2$ for any two input states $\rho_1, \rho_2$.
Then the following holds for any protocol $\Pi$:
\begin{align*}
 QIC(\Pi, \rho) \geq p QIC(\Pi, \rho_1) + (1-p) QIC(\Pi, \rho_2).
\end{align*}
\end{lemma}

\begin{proof}
Consider purifications $\rho_i^{A_{in} B_{in} R}$ of each $\rho_i$.
By introducing a new selector reference subsystem $S$, we have the purification 
$\ket{\rho}^{A_{in} B_{in} R S} = \sqrt{p} \ket{\rho_1}^{A_{in} B_{in} R} \ket{1}^S + \sqrt{1-p} \ket{\rho_2}^{A_{in} B_{in} R} \ket{2}^S$. Also consider the state $\hat{\rho}^{A_{in} B_{in} R \hat{S}} =\Delta_S (\rho^{A_{in} B_{in} R S})$ obtained if
the $S$ selector system was passed through a measurement channel $\Delta^{S \rightarrow \hat{S}}$ to obtain classical system $\hat{S}$.
We must verify that the quantum information cost satisfies the concavity property.
We consider information quantities taken with respect to the protocol $\Pi$ run on different inputs:
$\rho^{A_{in} B_{in} RS}, \hat{\rho}^{A_{in} B_{in} R \hat{S}}, \rho_2^{A_{in} B_{in} R}, \rho_2^{A_{in} B_{in} R}$ , and identify with respect to which state we are considering by a subscript on the corresponding conditional mutual informations.
\begin{align*}
		QIC (\Pi, \rho) 
		&= I(C_1 ; R S | B_0)_\rho + I(C_2 ; R S | A_1)_\rho + \cdots \\
		& \geq I(C_1 ; R \hat{S} | B_0)_{\hat{\rho}} + I(C_2 ; R \hat{S} | A_1)_{\hat{\rho}} + \cdots \\
		& = I(C_1 ; \hat{S} | B_0)_{\hat{\rho}} + I(C_1 ; R | \hat{S} B_0)_{\hat{\rho}} + I(C_2 ; \hat{S} | A_1)_{\hat{\rho}} + I(C_2 ; R | \hat{S} A_1)_{\hat{\rho}} + \cdots \\
		& \geq I(C_1 ; R | \hat{S} B_0)_{\hat{\rho}} + I(C_2 ; R | \hat{S} A_1)_{\hat{\rho}} + \cdots \\
		& = p (I(C_1 ; R | B_0)_{\rho_1} + I(C_2 ; R | A_1)_{\rho_1} + \cdots) \\
			& + (1-p) (I(C_1 ; R | B_0)_{\rho_2} + I(C_2 ; R | A_1)_{\rho_2} + \cdots) \\
		& = p QIC(\Pi, \rho_1) + (1-p) QIC(\Pi, \rho_2).
\end{align*}
The first equality is by definition of quantum communication cost of $\Pi$ on input $\rho$,
the first inequality uses the data processing inequality for conditional mutual information,
the second equality is by the chain rule for conditional mutual information,
the second inequality is by non-negativity of conditional mutual information,
the third equality uses a convex rewriting of conditional mutual information,
and the last equality is from definition.
\end{proof}

\subsection{Proof of Theorem \ref{th:qic}}
	\label{sec:qicam}

We now have all the tools to prove Theorem \ref{th:qic}, that states that
for any channel $\N$, any input state $\rho$ and any error parameter $\epsilon \in (0, 2]$,
$QIC (\N, \rho, \epsilon) = AQCC (\N, \rho, \epsilon)$. At $\epsilon = 2$, both quantities vanish.
The interesting regime is for $\epsilon \in (0, 2)$. Note that keeping track of rounds, we get the result for bounded rounds.

For the direct coding part, to prove the $\geq$ direction, we take an arbitrarily small $\delta > 0$,
and use Corollary \ref{cor:cty} to find an $0 < \epsilon^\prime < \epsilon$ such that
$QIC (\N, \rho, \epsilon - \epsilon^\prime) \leq QIC (\N, \rho, \epsilon) + \delta$.
We then consider a protocol $\Pi \in \T (\N, \rho, \epsilon - \epsilon^\prime)$ 
satisfying $QIC (\Pi, \rho) \leq QIC (\N, \rho, \epsilon - \epsilon^\prime) + \delta$.
We now use Lemma \ref{lem:coding} to find, for any sufficiently large $n$, a protocol $\Pi_n \in \T (\Pi^{\otimes n}, \rho^{\otimes n}, \epsilon^\prime)$ satisfying $\frac{1}{n} QCC (\Pi_n) \leq QIC(\Pi, \rho) + \delta$. We then have the following chain of inequality:
\begin{align*}
\frac{1}{n} QCC (\Pi_n) & \leq QIC(\Pi, \rho) + \delta \\
		& \leq QIC (\N, \rho, \epsilon - \epsilon^\prime) + 2 \delta \\
		& \leq QIC (\N, \rho, \epsilon) + 3 \delta.
\end{align*}
Since $\delta> 0$ is arbitrarily small and this holds for all sufficiently large $n$, we only have to verify that $\Pi_n \in \T_n (\N^\otimes, \rho^{\otimes n}, \epsilon)$ to complete the proof of the $\geq$ direction. We have for each $i \in [n]$,
\begin{align*}
 \| \Tra{\neg A_{out}^i B_{out}^i} \Pi_n (\rho^{\otimes n}) - \N (\rho) \| & \leq
 \| \Tra{\neg A_{out}^i B_{out}^i} \Pi_n (\rho^{\otimes n}) - \Pi (\rho) \| \\
	& +	\| \Pi (\rho) - \N (\rho) \| \\
	& \leq \| \Pi_n (\rho^{\otimes n}) - \Pi^{\otimes n} (\rho^{\otimes n}) \|  + \epsilon - \epsilon^\prime \\
	& \leq \epsilon,
\end{align*}
in which we first use the triangle inequality, and then monotonicity of the trace distance under partial trace.

For the converse part, to prove the $\leq$ direction, we combine Corollary \ref{cor:add} and Corollary \ref{cor:qicnqccn}, and get the result since the following holds for all $n$:
\begin{align*}
QIC (\N, \rho, \epsilon) &= \frac{1}{n} QIC_n (\N^{\otimes n}, \rho^{\otimes n}, \epsilon) \\
	&\leq \frac{1}{n} QCC_n (\N^{\otimes n}, \rho^{\otimes n}, \epsilon).
\end{align*}

\section{Quantum Information Complexity of Functions}

\subsection{Error for Classical Functions and Inputs}
\label{sec:cltd}

A potential application of the quantum information complexity paradigm is to prove quantum communication 
complexity lower bounds for classical functions.
Hence, we want to make sure that the notion of distance we use for 
quantum channels and states is also an interesting one in such context.
First, remember that for classical states $\rho^{A B} = \sum_{x, y} p_{XY} (x, y) \kb{x}{x}^A \otimes \kb{y}{y}^B$,
we can use a canonical basis $\{ \ket{xy} \}$ to obtain a purification 
$\ket{\rho}^{A B R} = \sum_{x, y} \sqrt{p_{XY} (x, y)} \ket{x}^A \ket{y}^B \ket{xy}^R$.
Then, for $\Delta_R$ the measurement channel in the $\ket{xy}$ basis on system $R$, 
we have
\begin{align*}
 \| (\Pi \otimes \Delta_R) (\rho) - (\N \otimes \Delta_R) (\rho) \|_{ABR} \leq
	\| \Pi (\rho) - \N (\rho) \|_{ABR}
\end{align*}
by the monotonicity of the trace distance under noisy channels.
For classical functions, we consider channels $\N$ such that we have
\begin{align}
\N (\kb{x}{x}^{A_{in}} \otimes \kb{y}{y}^{B_{in}}) = \kb{f_A (x,y)}{f_A (x, y)}^{A_{out}} \otimes \kb{f_B (x, y)}{f_B (x, y)}^{B_{out}}
\end{align}
 and then on classial inputs we have
\begin{align}
(\N \otimes \Delta_R) (\rho^{ABR}) = \sum_{x, y} p_{XY} (x, y) \kb{f_A(x, y)}{f_A(x, y)}^{A_{out}} \otimes \kb{f_B (x, y)}{f_B (x, y)}^{B_{out}} \otimes \kb{xy}{xy}^R. 
\end{align}
We get
\begin{align*}
 \| (\Pi \otimes \Delta_R) (\rho) &- (\N \otimes \Delta_R) (\rho) \|_{A_{out} B_{out} R} \\
	&= 	\| \sum_{x, y} p_{XY} (x, y) \kb{xy}{xy}^R \otimes (\Pi(\kb{x}{x} \otimes \kb{y}{y}) \\
	& - \kb{f_A (x, y)}{f_A (x, y)} \otimes \kb{f_B (x, y)}{f_B (x, y)} ) \|_{A_{out} B_{out} R} \\
	&= 	\sum_{x, y} p_{XY} (x, y) \| \Pi(\kb{x}{x} \otimes \kb{y}{y}) \\
	&-  \kb{f_A(x, y)}{f_A(x, y)} \otimes \kb{f_B (x, y)}{f_B (x, y)} \|_{A_{out} B_{out}},
\end{align*}
so if we further apply the measurement channel $\Delta_{AB}$ in the output basis,
\begin{align*}
	 \| (\Delta_{AB} \circ \Pi \otimes \Delta_R) (\rho) &- ( \Delta_{AB} \circ \N \otimes \Delta_R) (\rho) \|_{A_{out} B_{out} R} \\
	&= 	\sum_{x, y} p_{XY} (x, y) \| \sum_{z_A, z_B} p_{Z_A Z_B | \Pi (x, y)} (z_A, z_B | \Pi (x, y)) \kb{z_A}{z_A} \otimes \kb{z_B}{z_B} \\
	& - \kb{f_A (x, y)}{f_A (x, y)} \otimes \kb{f_B (x, y)}{f_B (x, y)}  \|_{A_{out} B_{out}}  \\
	&= 	\sum_{x, y} p_{XY} (x, y)  ((1 - p_{Z_A Z_B | \Pi (x, y)} (z_A = f_A (x, y), z_B = f_B (x, y) | \Pi (x, y))) \\
	& + \sum_{(z_A, z_B) \not= (f_A (x, y), f_B (x, y))} p_{Z_A Z_B | \Pi (x, y)} (z_A, z_B | \Pi (x, y)))  \\
	&=	\sum_{x, y} p_{XY} (x, y) (2 Pr[\Delta_{AB} \circ \Pi (x, y) \not= (f_A(x,y), f_B (x, y))]),
\end{align*}
and so applying monotonicity of the trace distance one last time, we get the following result.
\begin{lemma}
	For classical functions $f_A, f_B$, a channel $\N (\kb{x}{x} \otimes \kb{y}{y}) = \kb{f_A(x,y)}{f_A (x, y)} \otimes \kb{f_B (x, y)}{f_B (x, y)}$ 
	and a classical input state $\rho^{A_{in} B_{in}} = \sum_{x, y} p_{XY} (x, y) \kb{x}{x}^A \otimes \kb{y}{y}^B$,
	the following holds for any $\Pi (\N, \rho, \epsilon)$:
\begin{align*}
	\sum_{x, y} p_{XY} (x, y) (Pr[\Delta_{AB} \circ \Pi (x, y) \not= (f_A (x, y), f_B (x, y)) ]) \leq \frac{\epsilon}{2}
\end{align*}
\end{lemma}

\subsection{Reduction of Disjointness to AND}

In this section, we would like to establish a relationship between the quantum information cost of protocols
computing the disjointness function and those computing the AND of two bits.
The disjointness function is defined as $DISJ_n (x, y) = \neg (\bigvee_{i \in [n]} x_i \wedge y_i)$,
for bit strings $x = x_1 \cdots x_n, y = y_1 \cdots y_n$ and $x_i \wedge y_i$ the AND of the two bits $x_i, y_i$.
Given any protocol $\Pi_D$ computing $DISJ_n$ with error at most $\epsilon$ on all $n$-bit input, 
we can also use it to compute AND with error at most $\epsilon$ on all single bit input by 
setting $n-1$ inputs to $0$ on one side. Consider any
distribution $\mu$ on $00, 01,10$, and the corresponding state 
$\sigma_\mu =  \mu (00) \kb{00}{00} + \mu (01) \kb{01}{01} + \mu (10) \kb{10}{10}$.
We define in this way $n$ different protocols $\Pi_i$ for AND on an arbitrary input $\rho$, by setting, for
$i \in [n]$, the $i$-th input to $\Pi_D$ to the input $\rho$ of the AND instance, and the
$n-1$ remaining inputs to $\sigma_\mu^{\otimes n-1}$. We further combine them in a protocol $\Pi_A$ that is their average.
Running the average protocol on input $\sigma_\mu$, we get the following result.

\begin{theorem}
For any $M$-message protocol $\Pi_D$ computing $DISJ_n$ with error $\epsilon \in [0, 2]$ on all inputs, there exists a $M$-message protocol $\Pi_A$ computing
AND with error $\epsilon$ on all inputs and satisfying the following:
\begin{align*}
QIC(\Pi_A, \sigma_\mu) = \frac{1}{n} QIC(\Pi_D, \sigma_\mu^{\otimes n}),
\end{align*}
for the state $\sigma_\mu =  \mu (00) \kb{00}{00} + \mu (01) \kb{01}{01} + \mu (10) \kb{10}{10}$, and any probabitlity distribution $\mu$ on $\{00, 01, 10 \}$.
\end{theorem}

\begin{proof}

We first notice that we can replace the quantum criteria of error by the classical one for worst case input and obtain the same result.
We fix $\mu$ and simply write $\sigma$ instead of $\sigma_\mu$.
We start by defining the $n$ protocols $\Pi_i$ discussed above, and then use them to define the average protocol $\Pi_A$.

\hrulefill

Protocol $\Pi_i$ on input $\rho$: \\
-Let $\psi_D$ be the entangled state used in $\Pi_D$, and consider a purification
 $\sigma_j^{A_{in}^j B_{in}^j R^j}$ to $\sigma_j$, the $j$-th $\sigma$ input to $\Pi_D$.
The entangled state for the protocol will be $\psi_D \otimes \sigma^{\otimes n-1}$, with $\sigma$'s for all
$j$ different than $i$. The registers $T_A, A_{in}^1, \cdots, A_{in}^{i-1}, A_{in}^{i+1}, \cdots A_{in}^n, R^1,
\cdots R^{i-1}$ are given to Alice, and $T_B, B_{in}^1, \cdots, B_{in}^{i-1}, B_{in}^{i+1}, \cdots B_{in}^n, R^{i+1},
\cdots R^n$ to Bob. \\
-Using the $\sigma_j$ states given as pre-shared entanglement to simulate 
the other inputs, run protocol $\Pi_D$ on input $\sigma_1 \otimes \cdots \otimes \sigma_{i-1} \otimes \rho \otimes \sigma_{i+1} \otimes \cdots \otimes \sigma_n$. \\
-Take as output the $A_{out}, B_{out}$ registers of protocol $\Pi_D$.

\hrulefill

From these protocols, we define a new one for AND that is their uniform average.

\hrulefill

Protocol $\Pi_A$ on input $\rho$: \\
-The entangled state $\psi_A$ contains many parts: it contains the $\psi_D^{T_A T_B}$ entangled state for $\Pi_D$, as well as $2n$ $\sigma$ states,
with the $R^j$ register given to Alice for the first $n$ of them, and the $R^j$ register given to Bob for the $n$ last.
Denote $D_A = A_{in}^1 \otimes \cdots \otimes A_{in}^{2n} \otimes R^1 \otimes \cdots \otimes R^n, D_B = B_{in}^1 \otimes \cdots \otimes B_{in}^{2n} \otimes R^{n+1} \otimes \cdots \otimes R^{2n}$ that correspond to Alice's and Bob's share of theses states, respectively. The entangled state also contains some padding pure states $\phi^{A_{in}^{\otimes (n-1)} B_{in}^{\otimes (n-1)}}$ to be swapped with the selected $\sigma$'s, held in registers $P_A, P_B$, and it also contains selector 
registers in state $\ket{\theta} = \sum_i \frac{1}{\sqrt{n}} \ket{i}^{S_A} 
\ket{i}^{S_B}$. \\
-Coherently control what to input into the $\Pi_D$ protocol: on control set 
to $i$, input state $\rho$ into registers $A_{in}^i B_{in}^i$, use $\sigma_{1}, \cdots, \sigma_{i-1}$ in registers $A_{in}^j B_{in}^j$ for $j < i$, and use $\sigma_{n + i +1}, \cdots, \sigma_{2n}$ in registers $A_{in}^j B_{in}^j$ for $j > i$.
Swap the padding pure states in $P_A, P_B$ into the $D_A, D_B$ registers that are used as input to $\Pi_D$. \\
-Run protocol $\Pi_D$ on given inputs until it has finished. \\
-Take as output the $A_{out}, B_{out}$ registers of protocol $\Pi_D$.

\hrulefill

We first verify that protocol $\Pi_A$ successfully accomplish its task. 
It is clear that each $\Pi_i$ as well as $\Pi_A$ are also $M$ messages protocols.
Note that by 
the structure of the above protocol and because the selector registers are traced 
out at the end, the output is of the form $\Pi_A (\rho) = \sum_i \frac{1}{n} 
\rho_{out}^i$ for $\rho_{out}^i = \Pi_i (\rho)$ the output of protocol $\Pi_i$.
Further, each $\Pi_i$ satisfies $\Pi_i (\rho) = \Pi_D (\sigma_1 \otimes \cdots \otimes \sigma_{i-1} \otimes \rho \otimes \sigma_{i+1} \cdots \otimes \sigma_n)$.
Then, the result follows from the triangle inequality for the trace distance
and the structure of the different protocols and channels, using the completely classical channels
$\N_A = AND \circ \Delta_{A_{in} B_{in}}, \N_D = DISJ_n \circ \Delta_{A_{in} B_{in}}$ to ensure that indeed,
for any classical input $\rho$ and any $i \in [n]$, 
$\N_A (\rho) = \N_D (\sigma_1 \otimes \cdots \otimes \sigma_{i-1} \otimes \rho \otimes \sigma_{i+1} \cdots \otimes \sigma_n)$:
\begin{align*}
\| \Pi_A (\rho) - \N_A (\rho) \|_{A_{out} B_{out} R} &= \| \sum_i \frac{1}{n} \rho_{out}^i  - \sum_i \frac{1}{n} \N_A (\rho) \|_{A_{out} B_{out} R} \\
& \leq \sum_i \frac{1}{n} \|  \rho_{out}^i  -  \N_A (\rho) \|_{A_{out} B_{out} R} \\
& = \sum_i \frac{1}{n} \|  \rho_{out}^i  -  
\N_D (\sigma_1 \otimes \cdots \otimes \sigma_{i-1} \otimes \rho \otimes \sigma_{i+1} \cdots \otimes \sigma_n) \|_{A_{out} B_{out} R} \\
&\leq \sum_i \frac{1}{n} \epsilon \\
&= \epsilon.
\end{align*}
With a classical error parameter, the result follows by a similar average argument.
We must now verify that the quantum information cost satisfies the desired property on input $\rho = \sigma$.
First, note that if Alice's registers are traced out, then Bob's selector register is effectively a classical register,
and similarly for Alice's selector register if Bob's registers are traced out.
Then, conditional on some classical state of the selector register,
the protocol $\Pi_A$ acts as the corresponding protocol $\Pi_i$ running $\Pi_D$ on input $\sigma_1 \otimes \cdots \otimes \sigma_{i-1} \otimes \sigma \otimes \sigma_{n+i+1} \otimes \cdots \otimes \sigma_{2n}$, and
the other $\sigma_j$'s and $\phi$ are left untouched, in product form to $R$ and all $C_i$ registers.
Then if we run $\Pi_A$ on input $\sigma$, $A_{in} B_{in} R$ act as $A_{in}^i B_{in}^i R^i$ in $\Pi_D$
on input $\sigma^{\otimes n}$, so
\begin{align*}
		QIC (\Pi_A, \sigma)
		&= I(C_1 ; R | B_{in} T_B P_B D_B S_B) + I(C_2 ; R | A_1 D_A S_A) + I(C_3 ; R | B_2 D_B S_B) +\cdots \\
		&= \frac{1}{n} ( I(C_1 ; R | B_{in} T_B P_B D_B (S_B = 1) ) + I(C_2 ; R | A_1 D_A (S_A = 1)) \\
		&+ I(C_3 ; R | B_2 D_B (S_B = 1)) + \cdots \\
		&+ \cdots \\
		& +  I(C_1 ; R | B_{in} T_B P_B D_B (S_B = n) ) + I(C_2 ; R | A_1 D_A (S_A = n)) \\
		&+ I(C_3 ; R | B_2 D_B (S_B = n)) + \cdots ) \\
		&= \frac{1}{n} ( I(C_1 ; R^1 | B_{in}^1 (B_{in}^2 \cdots B_{in}^n) T_B (S_B = 1) ) + I(C_2 ; R^1 | A_1 R^2 \cdots R^n (S_A = 1)) \\
		&+ I(C_3 ; R^1 | B_2 (S_B = 1)) + \cdots \\
		&+ \cdots \\
		& +  I(C_1 ; R^i | (B_{in}^1 \cdots B_{in}^{i-1}) B_{in}^i (B_{in}^{i+1} \cdots B_{in}^n) T_B R^1 \cdots R^{i-1}  (S_B = i) ) \\
		& + I(C_2 ; R^i | A_1 R^{i+1} \cdots R^n (S_A = i)) \\
		& + I(C_3 ; R^i | B_2 R^1 \cdots R^{i-1} (S_B = i)) + \cdots ) \\
		&= \frac{1}{n} ( I(C_1 ; R^1 \cdots R^n | B_{in}^1  \cdots B_{in}^n T_B ) + I(C_2 ; R^1 \cdots R^n | A_1 ) \\
		&+ I(C_3 ; R^1 \cdots R^n | B_2 ) + \cdots ) \\
		& = \frac{1}{n} QIC (\Pi_D, \sigma^{\otimes n}).
\end{align*}
The first equality is by definition of quantum information cost of $\Pi^A$,
the second equality uses the above remark about the selector register of one party being
classical when the registers of the other party are traced out, along with a convex rewriting of conditional mutual information,
the third equality uses the above remark about the product structure of the untouched $\sigma$'s in
$D_A, D_B$, and about $A_{in} B_{in} R$ acting as $A_{in}^i B_{in}^i R^i$
 depending on 
the classical state of the selector register, along with 
 both the facts that mutual information in 
product systems is zero, and that conditioning on a product system is 
useless, the next equality follows by noticing that these quantities are the same as those in $\Pi_D$, and by
recursively applying the chain rule from top to bottom for odd $C_i$ terms, and from bottom up for even $C_i$ terms, 
and then the last equality is by definition of the quantum information cost of $\Pi_D$.
\end{proof}

Then $\inf_{\Pi_A \in \T^M (AND, \epsilon)} QIC (\Pi_A, \sigma_\mu)
\leq \frac{1}{n} \inf_{\Pi_D \in \T^M (DISJ_n, \epsilon) } QIC (\Pi_D, \sigma_\mu^{\otimes n})$, in which the
infima are taken over $M$-message protocols with worst-case error $\epsilon$ for classical inputs, on the respective functions.
Then, for any $\mu$, $\inf_{\Pi_D \in \T^M (DISJ_n, \epsilon) } QIC (\Pi_D, \sigma_\mu^{\otimes n}) \leq QCC^M (DISJ_n, \epsilon)$, and
a lower bound on $\inf_{\Pi_A \in \T^M (AND, \epsilon)} QIC (\Pi_A, \sigma_\mu)$, for any distribution $\mu$ on $00, 01, 10$, implies 
a lower bound on the quantum communication complexity of any $M$-message protocol $\Pi_D$ computing $DISJ_n$
with worst case error $\epsilon$. Is there a $\mu$ for which we can we get an interesting lower bound 
for $\inf_{\Pi_A \in \T^M (AND, \epsilon)} QIC (\Pi_A, \sigma_\mu)$?
In particular, note that this lower bound would need to have a dependence on the number of rounds,
since there are known protocols for computing $DISJ_n$ at communication cost $O(\frac{n}{M} + M)$ for
$M$-message protocols, and in particular letting $M \in \theta (\sqrt{n})$ in these we get an optimal protocol of $\Theta (\sqrt{n})$ \cite{AA03}.
Note that similar techniques
were used in \cite{JRS03}, adapted from the classical result of \cite{BYJKS02}, to obtain a lower bound of $\Omega (\frac{n}{M^2} + M)$ for bounded round protocols.
However, with their different notion of quantum information cost,
they started with a result that we would restate as
$\inf_{\Pi_A \in \T^M (AND, \epsilon)} QIC^\prime (\Pi_A, \sigma_\mu) \leq \frac{M}{n} \inf_{\Pi_D \in \T^M (DISJ_n, \epsilon) } QCC (\Pi_D)$,
with the infima also taken with respect to $M$-message protocols. It is interesting to 
note that we seem to start with a factor of $M$ less
here. We suspect that the reason for this is the following.
In our notation, their definition
of their quantum information cost would be
\begin{align}
QIC^\prime (\Pi, \rho) = \sum_{i>0, odd}  I(X; B_{i-1} | D ) + \sum_{i>0, even} I(Y; A_{i-1} | D),
\end{align} 
for protocols that keep local copies $X, Y$ of Alice's and Bob's classical inputs, and in which it is requested of $\mu$ that, 
for some random variable $D$ inaccessible to Alice and Bob and
conditional on the event $D=d$, the classical random variables $X_d = X | (D =d), Y_d = Y | (D=d)$ are independent.
Note that if we take the uniform distribution on Alice's input and the $0$ input for Bob, a product distribution, then any protocol that starts by distributing copies of all of Alice's bit, and then simply exchange a dummy qubit in the  $\ket{0}$ state for all remaining messages, would have $QIC^\prime \in \Omega (M \cdot n)$, for $M$ messages and $n$ bits inputs.
The communication is of the order of $M + n$, so for $M \in \omega (1)$, since we are interested in the regime $M \leq n$, this is a factor of $M$ more than the communication.
In contrast to their notion, our definition of quantum information cost is bounded by the communication.
If we could get a bound proportional to theirs for 
$\inf_{\Pi_A \in \T^M (AND, \epsilon)} QIC (\Pi_A, \sigma_\mu) \in \Omega (\frac{1}{M})$, 
we would get a lower bound of $\Omega (\max (M, \frac{n}{M}))$ on communication, thus
matching the best known upper bound for bounded round quantum protocol
 for $DISJ_n$.
We conjecture that this is the case, and so that this is worth looking into.
Note however that this would probably require new techniques to lower bound the conditional quantum mutual information, 
a quantity which is notoriously hard to lower bound \cite{LR73, BCY11}. We further discuss these issues in the conclusion.

\section{Conclusion}

We have defined a new notion of quantum information cost and a corresponding
notion of quantum information complexity.
In contrast to previously defined notions, these directly 
provide a lower bound on the communication, independent of round complexity.
To define the quantum information cost of a protocol on an input quantum 
state, we take a detour through classical information cost and 
provide a different perspective on it, relating it to noisy 
channels simulation with side information at the receiver, a variant 
of the classical reverse Shannon theorem studied in information theory. 
This provides a different proof that the information complexity is an achievable 
rate for amortized communication complexity, one that preserves the round complexity, 
and leads to an easier 
quantum generalization than the transcript based view on classical 
information cost. Using this quantum generalization,
we provide an operational interpretation for the $\epsilon$-error 
quantum information complexity of a channel $\N$ on input $\rho$
as the $\epsilon$-error amortized quantum communication complexity 
of $\N$ on input $\rho$, and in this sense provides the right quantum 
generalization of the classical information complexity. We prove some 
interesting properties of quantum information complexity, like additivity 
and convexity. We also prove that on channels implementing, with $\epsilon$ error, 
through a protocol $\Pi$, a classical function $f$ on a classical input 
distributed according to $\mu$, the error criterion that we use provides
an upper bound on the average probability of failure,
$Pr_\mu [\Pi (x, y) \not = f(x, y)] \leq  \frac{\epsilon}{2}$, 
hence giving it an operational meaning in the context of quantum 
information complexity of classical functions. 

An important application of classical information complexity is to prove 
communication complexity lower bounds, and it is reasonable to hope 
that quantum information complexity will lead to interesting lower bound 
on quantum communication complexity. In \cite{JRS03}, the authors use 
a different quantum generalization of information cost to derive an elegant proof of a lower 
bound on the bounded round quantum communication complexity of set disjointness. 
For $M$-message protocols, the authors prove a lower bound of $\Omega (\frac{n}{M^2} + M)$ on 
the quantum communication complexity of $DISJ_n$, close to the best known upper 
bound of $O(\frac{n}{M} + M)$. However, the notion of quantum information cost that they 
use can be as high as $\Omega (M \cdot C)$ for $M$-message protocols communicating $C$ qubits. 
Could we use our notion of quantum information cost, which does not have this 
possible dependence on $M$, to close the gap and prove a lower bound of
 $\Omega (\frac{n}{M} + M)$ on the quantum communication complexity?
We are able to show that on ($n$-fold tensor) classical inputs with no 
support for $x=y=1$, the quantum communication cost of any good protocol (with low error on all input) for $DISJ_n$ is 
at least $n$ times the quantum information cost for any good protocol for 
the AND function on the same input. In contrast, in \cite{JRS03} they obtain a
bound of at least $\frac{n}{M}$ times for a related quantity. However, we fall short of proving a 
$\Omega(\frac{1}{M})$ lower bound on the quantum information cost of good protocols for 
the AND function. The reason for this is that, for the moment, we lack 
the necessary tools to lower bound the quantum information cost, 
which is defined as a sum of conditional mutual information with a quantum 
conditioning register. In \cite{JRS03} the conditioning system is 
classical, and conditioning on a classical system is equivalent to taking an average.
They are then able, through a clever combination of a local transition lemma and an average encoding theorem, to relate the sum of the square roots of all terms in their quantum information
cost quantity to the probability of success of the $M$-message protocol. The round dependent lower
bound follows by using a convexity argument to put these terms inside the square root.
Could we do something similar?
The conditional quantum mutual information has an history of being hard to lower 
bound \cite{LR73, BCY11}, but hopefully recent developments in quantum information theory 
will help in this task. In particular, Berta, Seshadreesan and Wilde lists many rewriting of the 
conditional mutual information in terms of a divergence quantity \cite{BSW14}; hopefully it will be possible 
to use such a rewriting to derive an analogue of the average encoding theorem \cite{KNTZ01} with a quantum 
conditioning register. Note that the purifying register $R$ can play the role of a 
classical system for us, since by the data processing inequality we can
measure it first and still obtain a valid lower bound.
Maybe another approach to lower bound quantum information complexity would be to try to extend to the quantum setting the powerful approach of Kerenidis, Laplante, 
Lerays and Roland 
through zero-error protocols \cite{KLLRX12}. 
Note that Chailloux and Scarpa define a notion of superposed information cost \cite{CS13} appropriate for entangled games with no communication to obtain an exponential decay result for parallel repetition of such games. Is there a connexion between this quantity and ours?

We are confident that this new notion of quantum information complexity will
stimulate interesting developments in quantum communication complexity, as well as in
quantum information theory to obtain tools that would prove helpful for such developments.
Interesting directions for this research program is first to obtain 
interesting lower bound on the quantum information complexity of specific functions, possibly
by developing further techniques for lower bounding the conditional mutual information. 
Also, with the operational interpretation that we prove for quantum information 
complexity, this opens the door to the study of direct sum questions in quantum 
communication complexity. A direct sum results shows that to implement $n$ 
instances of a task requires something on the order of $n$ times the resources
 needed to implement a single instance of the task. Hence, the ability to compress 
the best protocols down to something close to their quantum information cost would lead
to a direct sum result. 
A first step in this direction could be to try to obtain such a single-copy compression result for bounded 
round protocols, possibly by using a one-shot version of the state redistribution protocol. 
However, the one-shot result of \cite{YD09} does not seem to be sufficient for this 
purpose. Does there exist one-shot protocols for state redistribution, possibly interactive, that would 
enable us to obtain interesting compression results for the bounded round scenario? 
Note that a quantum generalisation of the correlated sampling protocol of Braverman and Rao \cite{BR11}
has been proved recently \cite{AJMSY14}.
In this quantum version, if Alice knows the spectral decomposition of a state $\rho$, Bob knows the spectral decomposition
of a state $\sigma$, and both know the relative entropy $D (\rho || \sigma)$ between these, then using shared entanglement, Bob can sample a state $\tilde{\rho}$ close to $\rho$ at communication cost proportional to $D(\rho || \sigma)$.
Can we further generalize  this to a setting with quantum side information, to compress messages while keeping quantum correlation?
Another interesting research direction would be to try to obtain the quantum 
analogue of the result about a prior-free information cost and its relation to 
worst-case amortized communication complexity. A first approach to try for this 
problem could be to obtain a composable quantum reverse Shannon theorem on arbitrary input,
with feedback
and side information at the receiver that would be tailor-made for our purpose. However, even
in the case without side information at the receiver, it is known that for the arbitrary
input case, standard entanglement, such as maximally entangled states, is not sufficient, and more exotic
forms of entanglement such as entanglement embezzling states \cite{vDH03} are needed \cite{BDHSW09,BCR11}.
As is made clear from our direct coding theorem, those reverse Shannon theorems achieve a stronger, global
error parameter than the more local one that we require. Are these forms of entanglement also required
in this restricted scenario?

Finally, it would be interesting to see if the quantum information complexity 
paradigm would enable to shed some light on the question of equivalence between 
the model of Yao versus the model of Cleve-Buhrman for quantum communication 
complexity. Our communication model is closer to the Cleve-Buhrman one, 
and we use the pre-shared entanglement in a crucial way in many of our results.  
It would be possible in principle, even though the interpretation might not be the same, 
to define the quantum information complexity in the Yao model by restricting the infimum to be taken
over protocols where the state $\psi$ is a pure product state. But since we can distribute entanglement at no cost with our definition of quantum information cost, what would be an appropriate definition
in this setting? 
 Would it be any easier 
to relate the information complexity in these two models than it is to relate the 
corresponding communication complexities? And what about amortized communication complexity?

\paragraph{Acknowledgement}

We are grateful to Omar Fawzi for a stimulating discussion in the early stages of this work and for useful feedback on a previous version, as well as to Gilles Brassard and Alain Tapp for useful discussions and feedback during the write-up of this paper. We acknowledge financial support from a Fonds de Recherche 
Qu\'ebec-Nature et Technologies B2 Doctoral research scholarship.

\appendix

\section{External Information: Classical and Quantum}

The material in this section is more exploratory in nature, and is still work in progress.
Correspondingly, the presentation will be less formal, but ideas from the previous 
sections can be used to formalize the material here. It can safely be skipped  without affecting
understanding in the other sections.

The classical external information cost is defined as $IC^{ext} = I (XY; \Pi (X, Y))$, 
and is usually viewed as a measure
of how much information the protocol leaks about the input to an external observer, 
being an allowed observer or a potential eavesdropper. It is a natural quantity to use 
for example in the standard communication complexity setting with cryptographic 
consideration, and also in the simultaneous message passing model, in which Alice and Bob
must send simultaneously a single message to an external referee who must produce 
the output. We give two alternative
interpretations of this quantity, one in the standard communication complexity setting, 
and one in a generalization of the simultaneous message passing model. 
We then show that the quantum generalizations 
 give rise to two different 
quantities, even though they were
both characterized by a unique classical quantity,
 and so the external information 
cost has a dual role depending on the
interpretation we want to give it. 

Let us first define the generalization of the simultaneous message passing model that we consider. 
We consider the same definition for
classical protocols as given in section \ref{sec:bas}, but the difference is now in the output: 
instead of having as output of the protocol Alice and Bob outputting some function 
of the transcript and their local input, we now want to have an external referee $R$
generating the output to the protocol by computing some function of the transcript only. 
If we want to obtain the simultaneous
message model from this generalization, we restrict protocols to two messages, 
and require the following Markov condition: $M_2 | M_1 Y R = M_2 | Y R$, 
i.e.~the second message must be independent of the first message when conditioned on the input $Y$
and shared randomness $R$.

The important point to notice is that now, when compressing the protocol, we
must insure that the referee gets all the necessary information in the 
transcript. We define an alternate external information
cost quantity as $IC_\mu^{\prime ext} = I(M_1^R; X R^A | R^R) 
+ I(M_2^R ; Y M_1^B R^B | M_1^R R^R) + I(M_3^R ; X M_2^A M_1^A R^A | M_2^R M_1^R R^R) + \cdots +
I(M_N^R; Y M_{N-1}^B \cdots M_1^B R^B | M_{N-1}^R \cdots M_1^R R^R)$,
in which we distinguish between Alice's, Bob's and the referee's copy 
of the public randomness $R$ and the messages $M_i$. Note that this 
is easily seen to be equivalent to $IC^{ext}$, using a similar argument 
as for the internal information cost, along with a combination of the chain rule and the Markov condition on
messages versus the inputs: $M_i | X Y M_{i-1} \cdots R = M_i | Y M_{i-1} \cdots R$ for 
even $i$, and $M_i | X Y M_{i-1} \cdots R = M_i | X M_{i-1} \cdots R$ for odd $i$. Similar 
to the case for internal information cost, the external information cost can 
then be viewed as an asymptotically achievable cost for transferring the messages such 
that both the other player and the referee can get the information, in the 
limit of many protocols. We then get an operational interpretation for the 
external information complexity as an achievable amortized communication complexity 
in this setting with an external referee.
It is then clear that the external information
should be interesting in the simultaneous message passing model.
Note that this setting in the interactive case is interesting only when the entropy of $f$, for $f$ the output function, is much larger than
the internal information complexity, since otherwise Alice and Bob can simply compute $f$ at this cost
and send the result to the referee at small additional cost.

It is less straightforward to generalize this to the quantum setting, 
mostly due to the fact that we cannot send a copy of messages to both 
the other player and to a referee in the quantum setting. 
The quantum generalization of this setting would be the following:
given a quantum channel $\N \in \C(A_{in} \otimes B_{in}, R_{out})$
and input state $\rho \in \D (A_{in} \otimes B_{in})$, Alice and Bob are 
given input register $A_{in}, B_{in}$ at the outset of the protocol, respectively, 
and the referee should output register $R_{out}$ at the end of the protocol, 
which should be in state $\N(\rho)$ up to some small error. In the protocol,
at each time step, after applying their unitary, the players send a communication 
register to the other player, and another one to the
referee, and keep some quantum memory. They want to minimize the total 
communication cost, that is, the sum of the cost from player to referee, and from player to player. The players and the referee can share an arbitrary 
tripartite entangled state at the outset of the protocol. A reasonable 
external quantum information cost for this setting, with the operational 
interpretation of being an asymptotically achievable rate of total communication, would
follow in the same way as the standard quantum information cost quantity, from repeated application of state redistribution,
along with an appropriate partition of the side information and unavailable information in this setting.

If we instead want to take the eavesdropper view on the external information 
cost, we go back to the standard two-party 
communication complexity setting, and want to consider a passive 
eavesdropper, who has access to previous correlation, to all communication 
and to all garbage information from the protocol at the end (classically, this
would be the transcript), but is not allowed 
to alter the communication. However, this quantity, when evaluated in a quantum setting for classical functions and inputs
would be bounded by a constant due to the fact that a passive quantum eavesdropper
cannot listen to the quantum communication, and due to the presence of reversible computing in
the quantum world (see \cite{Bra12a} for a discussion), which would limit the amount of
information the passive eavesdropper would get at the end.

Note that this channel simulation view of information cost also leads rather straightforwardly to possibly interesting
definitions of information cost in a multipartite setting, in particular in a model with some external coordinator,
as an achievable rate of total asymptotic communication which lower bounds single-copy communication.
It is then possible to define appropriate quantum generalizations, by carefully adapting the model, and then once again by considering
repeated application of state redistribution. Of course, the interest of such information cost quantities find their interest in potential applications.
We leave such possible applications of these potential definitions for future work.


\begin{thebibliography}{9}

\bibitem {AA03} Scott Aaronson, and Andris Ambainis.
\textit{Quantum Search of Spatial Regions.}
Proceedings of the 44rd Annual IEEE Symposium on Foundations of 
	Computer Science (2003): 200-209.

\bibitem {ADHW09} Anura Abeyesinghe, Igor Devetak, Patrick Hayden, and Andreas Winter.
\textit{The mother of all protocols: Restructuring quantum information's family tree.}
Proceedings of the Royal Society of London. Series~A 
(2009): 2537-2563.

\bibitem {AJMSY14} Anurag Anshu, Rahul Jain, Priyanka Mukhopadhyay, Ala Shayeghi, and Penghui Yao.
\textit{A new operational interpretation of relative entropy and trace distance between quantum states.}
arXiv:quant-ph/1404.1366.

\bibitem {BYJKS02} Ziv Bar-Yossef, T. S. Jayram, Ravi Kumar, and D. Sivakumar.
\textit{An information statistics approach to data stream and communication complexity.}
Proceedings of the 43rd Annual IEEE Symposium on Foundations of 
	Computer Science (2002): 209-218.

\bibitem {BBCR10} Boaz Barak, Mark Braverman, Xi Chen, and Anup Rao.
\textit{How to compress interactive communication.}
Proceedings of the 42nd Annual ACM Symposium on Theory of 
	Computing (2010): 67-76.

\bibitem {BCGST02} Howard Barnum, Claude Crepeau, Daniel Gottesman, Adam Smith, and Alain Tapp.
\textit{Authentication of Quantum Messages.}
Proceedings of the 43rd Annual IEEE Symposium on Foundations of 
	Computer Science (2002): 449-458.

\bibitem {BBCJPW93} Charles H. Bennett, Gilles Brassard, Claude 
	Cr\'epeau, Richard Jozsa, Asher Peres, and William K. Wooters.
\textit{\mbox{Teleporting} an unknown quantum state via dual 
	classical and {E}instein--{P}odolsky--{R}osen channels.}
Physical \mbox{Review} Letters 70.13 (1993): 1895-1899.

\bibitem {BDHSW09} Charles H. Bennett, Igor Devetak, Aram W. Harrow, Peter W. Shor, and Andreas Winter.
\textit{Quantum Reverse {S}hannon Theorem.}
arXiv:quant-ph/0912.5537.

\bibitem {BSST02} Charles H. Bennett, Peter W. Shor, John A. Smolin, and 
	Ashish V. Thapliyal.
\textit{Entanglement-assisted capacity of a quantum channel and the 
	reverse {S}hannon theorem.}
 IEEE Transactions on Information Theory 48.10 (2002): 2637-2655.

\bibitem {BW92} Charles H. Bennett and Stephen J. Wiesner.
\textit{Communication via one- and two-particle operators on
 {E}instein--{P}odolsky--{R}osen states.}
Physical \mbox{Review} Letters 69.20 (1992): 2881-2884.

\bibitem {BCR11} Mario Berta, Matthias Christandl, Renato Renner. 
\textit{The Quantum Reverse Shannon Theorem based on One-Shot Information Theory.}
Commununications in Mathematical Physics 306 (2011): 579-615.

\bibitem {BSW14} Mario Berta, Kaushik Seshadreesan, and Mark M. Wilde.
\textit{Renyi generalizations of the conditional quantum mutual information.}
arXiv:quant-ph/1403.6102.

\bibitem {BCY11} Fernando G. S. L. Brandao, Matthias Christandl, Jon Yard. 
\textit{Faithful Squashed Entanglement.}
Commununications in Mathematical Physics 306 (2011): 805-830.

\bibitem {Bra12a} Mark Braverman.
\textit{Interactive information complexity.}
Proceedings of the 44th Annual ACM Symposium on Theory of 
	Computing (2012): 505-524.

\bibitem {Bra12b} Mark Braverman.
\textit{Coding for Interactive Computation: Progress and Challenges.}
Proceedings of the 50th Annual Allerton Conference on
	Communication, Control, and Computing (2012): 1914:1921.

\bibitem {BEOPV13} Mark Braverman, Faith Ellen, Rotem Oshman, Toniann Pitassi, and Vinod Vaikuntanathan.
\textit{Tight bounds for set disjointness in the message passing model.}
Proceedings of the 54th Annual IEEE Symposium on Foundations of 
	Computer Science (2013): 668-677.

\bibitem {BGPW13} 	Mark Braverman, Ankit Garg, Denis Pankratov, and Omri Weinstein.
\textit{From information to exact communication.}
Proceedings of the 45th Annual ACM Symposium on Theory of 
	Computing (2013): 151-160.

\bibitem {BM13} Mark Braverman, and Ankur Moitra.
\textit{An information complexity approach to extended formulations.}
Proceedings of the 45th Annual ACM Symposium on Theory of 
	Computing (2013): 161-170.

\bibitem {BR11} Mark Braverman, and Anup Rao.
\textit{Information equals amortized communication.}
Proceedings of the 52nd Annual IEEE Symposium on Foundations of 
	Computer Science (2011): 748-757.

\bibitem {CS13} Andr\'e Chailloux, and Giannicola Scarpa.
\textit{Parallel Repetition of Entangled Games with Exponential Decay
via the Superposed Information Cost.}
arXiv:quant-ph/1310.7787, presented at QIP2014.

\bibitem {CB97} Richard Cleve, and Harry Buhrman.
\textit{Substituting quantum entanglement for communication.}
Physical Review~A 56.2 (1997): 1201-1204.

\bibitem {CSWY01} A. Chakrabarti, Yaoyun Shi ; A. Wirth, Andrew C.-C. Yao.
\textit{Informational complexity and the direct sum problem for simultaneous message complexity.}
Proceedings of the 42nd Annual IEEE Symposium on Foundations of 
	Computer Science (2001): 270-278.

\bibitem {vDH03} Wim van Dam, and Patrick Hayden.
\textit{Universal entanglement transformations without communication.}
Physical Review~A 67.6 (2003): 060302(R).

\bibitem {DY08} Igor Devetak, and Jon Yard.
\textit{Exact cost of redistributing multipartite quantum states.}
Physical \mbox{Review} Letters 100.23 (2008): 230501.

\bibitem {Dieks82} Dennis Dieks.
\textit{Communication by EPR devices.}
Physics Letters~A 92.6 (1982):271-272.

\bibitem {JN13} Rahul Jain, and Ashwin Nayak.
\textit{The space complexity of recognizing well-parenthesized expressions in the streaming model: the Index function revisited.}
arXiv: cs.cc/1004.3165.

\bibitem {JRS03} Rahul Jain, Jaikumar Radhakrishnan, and Pranab Sen.
\textit{A lower bound for bounded round quantum communication complexity of set disjointness.}
Proceedings of the 44th Annual IEEE Symposium on Foundations of 
	Computer Science (2003): 220-229.

\bibitem {KLLRX12} Iordanis Kerenidis, Sophie Laplante, Virginie Lerays, J\'er\'emie Roland, and David Xiao.
\textit{Lower Bounds on Information Complexity via Zero-Communication Protocols and Applications.}
Proceedings of the 53rd Annual IEEE Symposium on Foundations of 
	Computer Science (2012): 500-509.

\bibitem {KNTZ01} Hartmut Klauck, Ashwin Nayak, Amnon Ta-Shma, David Zuckerman.
\textit{Interaction in quantum communication and the complexity of set disjointness.}
Proceedings of the 33rd Annual ACM Symposium on Theory of 
	Computing (2001): 124-133.


\bibitem {LR73} Elliott H. Lieb, and Mary Beth Ruskai.
\textit{Proof of the strong subadditivity of quantum-mechanical
entropy.}
Journal of Mathematical Physics, 14 (1973): 1938-1941.

\bibitem {LD09} Zhicheng Luo, and Igor Devetak.
\textit{Channel Simulation With Quantum Side Information.}
 IEEE Transactions on Information Theory 55.3 (2009): 1331-1342.

\bibitem {MV14} Laura Man\v{c}inska, Thomas Vidick.
\textit{Unbounded entanglement can be needed to achieve the optimal success probability.}
arXiv:quant-ph/1402.4145.

\bibitem {Wat13} John Watrous.
\textit{Theory of Quantum Information.}
Lecture notes from Fall 2013, 
	https://cs.uwaterloo.ca/~watrous/CS766/ (2013).

\bibitem {Wilde11} Mark M. Wilde.
\textit{Quantum Information Theory.}
Cambridge University Press (2013),
Preliminary version available as: arXiv e-print quant-ph/1106.1445.

\bibitem {WZ82} William K. Wootters, and Wojciech H. Zurek.
\textit{A single quantum cannot be cloned.}
Nature 299.5886 (1982): 802-803.

\bibitem {Yao93} Andrew C.-C. Yao.
\textit{Quantum circuit complexity.}
Proceedings of the 34th Annual IEEE Symposium on Foundations of 
	Computer Science (1993): 352-361.

\bibitem {YD09} Jon T. Yard, and Igor Devetak.
\textit{Optimal Quantum Source Coding With Quantum Side Information at the Encoder and Decoder.}
 IEEE Transactions on Information Theory 55.11 (2009): 5339-5351.


\end{thebibliography}
\end{document}